\documentclass[12pt]{article}
\usepackage{amsmath}
\usepackage{amssymb}
\usepackage{amsthm}
\usepackage{graphicx}
\usepackage{lineno}
\usepackage{float}
\usepackage[all]{xy}
\usepackage{mathrsfs}

\usepackage{xcolor}

\newtheorem{theorem}{Theorem}

\newtheorem{definition}{Definition}
\newtheorem{proposition}{Proposition}

\begin{document}

\title{\bf A free energy principle for generic quantum systems}

\author{{Chris Fields$^a$\footnote{Corresponding author at: 23 Rue des Lavandi\`{e}res, 11160 Caunes Minervois, FRANCE; {\it E-mail address}: fieldsres@gmail.com}, Karl Friston$^b$, James F. Glazebrook$^c,^d$ and Michael Levin$^e$}\\ \\
{\it$^a$ 23 Rue des Lavandi\`{e}res, 11160 Caunes Minervois, FRANCE}\\
{\it$^b$ Wellcome Centre for Human Neuroimaging, University College London,} \\
{\it London, WC1N 3AR, UK} \\
{\it$^c$ Department of Mathematics and Computer Science,} \\
{\it Eastern Illinois University, Charleston, IL 61920 USA} \\
{\it$^d$ Adjunct Faculty, Department of Mathematics,}\\
{\it University of Illinois at Urbana-Champaign, Urbana, IL 61801 USA}\\
{\it$^e$ Allen Discovery Center at Tufts University, Medford, MA 02155 USA}
}

\maketitle

{\bf Abstract} \\
The Free Energy Principle (FEP) states that under suitable conditions of weak coupling,  random dynamical systems with sufficient degrees of freedom will behave so as to minimize an upper bound, formalized as a variational free energy, on surprisal (a.k.a., self-information).  This upper bound can be read as a Bayesian prediction error. Equivalently, its negative is a lower bound on Bayesian model evidence (a.k.a., marginal likelihood). In short, certain random dynamical systems evince a kind of self-evidencing. Here, we reformulate the FEP in the formal setting of spacetime-background free, scale-free quantum information theory.  We show how generic quantum systems can be regarded as observers, which with the standard freedom of choice assumption become agents capable of assigning semantics to observational outcomes.  We show how such agents minimize Bayesian prediction error in environments characterized by uncertainty, insufficient learning, and quantum contextuality.  We show that in its quantum-theoretic formulation, the FEP is asymptotically equivalent to the Principle of Unitarity.  Based on these results, we suggest that biological systems employ quantum coherence as a computational resource and -- implicitly -- as a communication resource.  We summarize a number of problems for future research, particularly involving the resources required for classical communication and for detecting and responding to quantum context switches. 

\tableofcontents

\section{Introduction}

Since its introduction as a theory of brain function \cite{friston:05, friston:06, friston:07, friston:10}, the variational Free Energy Principle (FEP) has been extended into an explanatory framework for living systems at all scales \cite{friston:13, friston:17, ramstead:18, ramstead:19, kuchling:20}, and shown to characterize, in its most general form, all random dynamical systems that remain measurable, and hence identifiable as persistent and separable entities, over macroscopic times \cite{friston:19}.  To summarize, it is shown in \cite{friston:19} that any system that has a non-equilibrium steady state (NESS) solution to its density dynamics i) possesses an internal dynamics that is conditionally independent of the dynamics of its environment, and ii) will continuously ``self-evidence'' by returning its state to (the vicinity of) its NESS.  The FEP is the statement that any measurable, i.e. bounded and macroscopically persistent, system will behave so as to satisfy these requirements.  Self-organization, the FEP tells us, is not a rare, special case, but a ubiquitous feature of physical systems with sufficient dynamical stability to be called ``things.''  All ``things,'' in particular, self-organize Markov blankets (MBs; see \cite{clark:17} for an informal review) comprising ``sensory'' states that encode incoming information and thus mediate the influence of external states on internal states, and ``active'' states that encode outgoing information and thus mediate the influence of internal states on external states.  This partitioning of ``things'' into internal and MB states means that every ``thing'' can be construed as a certain kind of ``particle'' -- a particle that is in open exchange with external states via its MB.  In short, the MB of any such ``particle'' underwrites conditional independence between its internal states and the external states of its environment by localizing and thereby restricting information exchange; hence, the MB can be viewed as separating internal from external states, while mediating their exchange.

As noted in \cite{friston:19}, generalizing the FEP to characterize the behavior of all ``things'' substantialy weakens the traditional distinction between ``cognition'' and merely ``physical'' dynamics, and hence weakens the even deeper, pretheoretical \cite{scholl:00} distinction between ``agents'' and mere ``objects'' \cite{fgl:21}.  Treating physical interaction as information exchange -- ``observation'' of the environment followed by ``action'' upon it -- redescribes the ``mechanical'' process of returning to the NESS -- as a random global (a.k.a., pullback) attractor -- in terms of inference. This can be read as ``active inference'' \cite{friston:10, friston:13, friston:17, friston:19} in which the existential integrity of the MB, and hence of the ``self'' -- ``world'' distinction \cite{levin:20}, can be maintained in the face of environmental fluctuations by changing internal states (via sensory states: c.f., perception) or changing external states (via active states: c.f., action).  It refocuses the discussion, in other words, from abstract trajectories (or flows) in state space to the MB itself as a concrete locus of ``identity'' in the form of persistent measurability, and hence of ``self-evidencing'' via active inference to maintain that identity.

The idea that all physical systems, including the environment at large, can be considered ``observers'' that also act on their surroundings to ``prepare'' them for subsequent observations has become commonplace in quantum theory, largely replacing the ``wave-function collapse'' postulate of traditional quantum mechanics \cite{vonN:55} with interaction-induced decoherence (i.e., dissipation of quantum coherence) as the generator of classical information \cite{omnes:92, zurek:98, zurek:03, schloss:03, schloss:07}.\footnote{Variants of quantum theory that postulate a physical collapse mechanism also invoke interaction, e.g. with gravity or a ``noise'' field, to generate classicality; see \cite{bassi:13} for review.  We will not discuss here the question of ``interpretations'' of quantum theory; see \cite{landsman:07} for a thorough review and \cite{cabello:15} for a more recent taxonomy highlighting fundamental assumptions.  Both decoherence and entanglement, in particular, have different physical meanings in different interpretations, although their mathematical representations and observable effects are interpretation-independent.  Our general approach can be viewed as replacing the ``measurement problem'' -- that such interpretations are designed to solve -- with an explicit, quantum-informational theory of measurement.  This theory is, we show, just the theory of the FEP.}  Indeed while quantum theory was originally developed -- and is still widely regarded -- as a theory specifically applicable at the atomic scale and below, since the pioneering work of Wheeler \cite{wheeler:83}, Feynman \cite{feynman:82}, and Deutsch \cite{deutsch:85}, it has, over the past few decades, been reformulated as a scale-free information theory \cite{nielsen:00, hardy:01, fuchs:03, brassard:05, chiribella:11, masanes:11} and is increasingly viewed as a theory of the process of observation itself \cite{fuchs:13, grinbaum:17, mermin:18, muller:20, fg:20, fgm:21}.  This newer understanding of quantum theory fits comfortably with the generalization of the FEP, and hence of self-evidencing and active inference, to all ``things'' as outlined in \cite{friston:19}, and with the general view of observation under uncertainty as inference.

In what follows, we take the natural next step from \cite{friston:19}, formulating the FEP as a generic principle of quantum information theory.  We show, in particular, that the FEP emerges naturally in any setting in which an ``agent" or ``particle'' deploys quantum reference frames (QRFs), namely, physical systems that give observational outcomes an operational semantics \cite{aharonov:84, bartlett:07}, to identify and characterize the states of other systems in its environment.  This reformulation removes two central assumptions of the formulation in terms of random dynamical systems employed in \cite{friston:19}: the assumption of a spacetime embedding (or ``background'' in quantum-theoretic language) and the assumption of ``objective'' or observer-independent randomness.  It further reveals a deep relationship between the ideas of local ergodicity and system identifiability, and hence the idea of ``thingness'' highlighted in \cite{friston:19}, and the quantum-theoretic idea of separability, i.e., the absence of quantum entanglement, between physical systems. 

Any quantum system that can be distinguished from its environment over time can, therefore, be regarded as self-organizing and self-evidencing as described in \cite{friston:19}.  We then show that when the FEP is taken to an asymptotic limit, it drives systems away from separability towards entanglement, and hence towards a supraclassical statistical coupling between each ``thing'' and its environment -- between the observer and the observed.  In this, the FEP reproduces the Principle of Unitary, i.e. the Principle of Conservation of Information, which similarly drives all interacting systems asymptotically toward entanglement.  Hence the FEP is, in an important sense, an alternative statement of the Principle of Unitarity, the most fundamental principle of quantum theory.  It therefore applies to a much broader array of systems than would fall under an intuitive idea of ``thingness,'' e.g. to quantum fields, and applies in principle from the Planck scale to cosmological scales.  Formulating the FEP as a generic principle of quantum information theory thus substantially expands the range of systems to which ``cognitive'' or information-processing concepts reasonably apply.

We begin by reviewing in \S \ref{2} the basic principles of quantum theory from an information-theoretic perspective, limiting the formalism to focus on the physical meaning of the theory.  Using the category-theoretic \cite{adamek:04, awodey:10} formalism of Channel Theory -- developed by Barwise and Seligman \cite{barwise:97} to formalize the operational semantics of natural languages -- we develop a generic formal representation of QRFs and show how the noncommutativity of QRFs induces quantum contextuality \cite{kochen:67, mermin:93}, a nonclassical effect demonstrating the presence of entanglement between distinct physical degrees of freedom.  We develop in \S \ref{3} a generic, formal description of how one quantum system identifies another quantum system as a persistent entity -- a ``thing'' -- and measures, records, and compares its states by deploying specific sequences of QRFs.  This identification and measurement process depends critically on breaking thermodynamic symmetries, and therefore on system-specific flows of energy.  These sections together provide a representation of generic quantum systems as observers, or in the language of Gell-Mann and Hartle \cite{gmh:89} ``information gathering and using systems'' (IGUSs), that is free of scale and spacetime embedding (i.e. ``background'') dependent assumptions. It also treats all probabilities as observer-relative.  We then show in \S \ref{FEP} how the FEP emerges in this setting and analyze its asymptotic behavior; in particular, we consider how the FEP addresses the fundamental problem posed by quantum context switches.  We conclude in \S\ref{disc} by discussing the relevance of these results to a scale-independent understanding of biological systems as ``particles'' that interact with other ``particles,'' whether these are other organisms, ``objects,'' or an undifferentiated environment.  We consider in particular the circumstances in which this ``particle'' nature can break down, and suggest that well-designed experiments may be expected to detect quantum context switches or violations of the Bell \cite{bell:64} or Leggett-Garg \cite{emary:13} inequalities, any of which indicate entanglement, by macroscopic biological systems under ordinary conditions.

\section{Physical interaction as information exchange} \label{2}

\subsection{What is ``quantum''?}

When physical interaction is viewed as information exchange, why it is ``quantum'' becomes obvious: the fundamental quantum of information is one bit, one unit of entropy, that one system exchanges with another.  One bit, one quantum of information, is the answer to one yes/no question.  Planck's quantum of action $\hbar$ is then naturally regarded as the action (energy $\cdot$ time) required to obtain one bit via any physical interaction.  The energy required to irreversibly obtain one bit, i.e., to receive and irreversibly record one bit, is given by Landauer's Principle as $\ln 2~k_B T$, with $k_B$ Boltzmann's constant and $T$ temperature \cite{landauer:61, landauer:99, bennett:82}.  The (minimum) time to irreversibly obtain one bit is then $\hbar /\ln2~k_B T$, roughly 30 fs at 310 K.  For comparison, the thermal dissipation time (in 3d space) due to time-energy uncertainty is $\pi \hbar /2\ln 2~k_B T$ \cite{lloyd:00}, roughly 50 fs at 310 K.  These values define a minimal timescale for biologically-relevant, irreversible information processing, roughly the timescale of molecular-bond vibrational modes \cite{zwier:10} and an order of magnitude shorter than photon-capture timescales \cite{wang:94}.

Viewing all physical interaction as information exchange -- and the bit as the fundamental ``quantum'' of information -- has the immediate consequence that interaction discretizes the state spaces of all observable degrees of freedom.  Given a set of mutually-commuting, binary-valued observables, i.e., quantum operators implementing yes/no questions as described below, a discrete state space for any observed system is constructed by assigning a basis vector to each possible outcome (yes or no, +1 or -1) for each observable.  These are finite-dimensional Hilbert spaces.  While Hilbert spaces with either finite or infinite dimension are introduced {\em ad hoc} in traditional quantum mechanics \cite{vonN:55}, finite-dimensional Hilbert spaces emerge naturally from the process of observation in the information theory.  As Fuchs puts it, infinite dimensional Hilbert spaces are, from an information perspective, merely a ``useful artifice'' permitting computation with differential equations \cite{fuchs:10}.  Hence in what follows, all Hilbert space dimensions will be finite dimensional.

\subsection{Unitarity}

The fundamental axiom of quantum theory is unitarity, again introduced {\em ad hoc} in traditional treatments.  If $U$ is an isolated system, its time propagator is a unitary operator:

\begin{equation} \label{prop}
\mathcal{P}_U = e^{- (\imath / \hbar) H_U t},
\end{equation}
\noindent
where $H_U$ is the ``internal'' Hamiltonian (i.e. energy) operator satisfying the Schr\"{o}dinger equation:

\begin{equation} \label{Schrodinger}
\imath \hbar (\partial / \partial t) \vert U(t) \rangle = H_U \vert U(t) \rangle,
\end{equation}
\noindent
where $|U(t) \rangle$ is the time-dependent state of $U$.  When $H_U$ and hence $\mathcal{P}_U$ are time-invariant, solutions have the form:

\begin{equation} \label{phase-rotation}
\vert U(t) \rangle = e^{- \imath \varphi t} \vert U(0) \rangle,
\end{equation}
\noindent
where $\vert U(0) \rangle$ is an initial ($t = 0$) state.  This Eq. \eqref{phase-rotation} describes a phase rotation by $\varphi$ per unit time $t$ in the Hilbert space $\mathcal{H}_U$; the initial state $\vert U(0) \rangle$ is preserved ``up to'' this phase rotation.  The dimension $d = \dim(\mathcal{H}_U)$ is the number of basis vectors of $\mathcal{H}_U$ and is, by definition, the quantity (in bits) of {\em observable} information encoded by the state $\vert U(0) \rangle$.  The number $d$ is clearly invariant under the phase rotation given by Eq. \eqref{phase-rotation}; hence the phase rotation is not an observable.  Unitary evolution as defined by Eq. \eqref{prop} is, therefore, simply evolution over time that conserves observable information. Hence, the fundamental axiom of quantum theory is not {\em ad hoc} at all: it is the Principle that observable information, like energy, is neither created nor destroyed by physical processes.

The appearance of the ``imaginary'' unit $\imath$ in Eq. \eqref{prop} - \eqref{phase-rotation} and the use of Hilbert spaces over the complex field $\mathbb{C}$ are standard in quantum theory but are often considered a mere convenience; compelling arguments for their necessity have only been developed recently \cite{moretti:17, renou:21}.  They can, however, be given a straightforward interpretation: they emphasize that the phase rotation implemented by $\mathcal{P}_U$ is not an observable dynamics and that the ``external'' time $t$ in these equations is not an observable, clock-referenced time.\footnote{Imagine, for an example, listening to a constant tone that has no beginning or end.  Writing time as $\imath t$ with $t$ real allows Minkowski spacetime to be given a Galilean (++++) metric.  See \cite{baez:20} for a discussion of $\imath$ as a shorthand for converting classical to quantum observables, and \cite{kauffman:11} for a more intuitive view of ``rotation by $\imath$'' as a formal operation.}  Indeed, no classical information -- no observational outcome -- has yet been obtained in the setting defined by Eq. \eqref{prop} - \eqref{phase-rotation}.  Characterizing observation as a process generating classical outcomes as a result requires an additional assumption of separability as outlined below.

\subsection{Separability and holographic encoding}\label{holographic-1}

Let $U$ be an isolated system as above, and let $U = AB$ be a bipartite decomposition of $U$, i.e. the Hilbert space $\mathcal{H}_U = \mathcal{H}_A \otimes \mathcal{H}_B$.  At a fixed time $t$, any such bipartite decomposition can be characterized by an entanglement entropy:

\begin{equation} \label{entanglement-entropy}
\mathcal{S} (\vert AB \rangle) = - \sum_i \vert \alpha_i \vert^2 \log_2 (\vert \alpha_i \vert^2 ),
\end{equation}
\noindent
where the coefficients $\alpha_i$ are the Schmidt coefficients given by:

\begin{equation} \label{schmidt}
\vert AB \rangle = \sum_i^m \alpha_i \vert u_i \rangle_A \vert v_i \rangle_B,
\end{equation}
\noindent
where the label $B$ is assigned so that $m = \dim(\mathcal{H}_B ) \leq \dim(\mathcal{H}_A )$ and the $\vert u_i \rangle_A$ and $\vert v_i \rangle_B$ are orthonormal states of $A$ and $B$ respectively.  The entanglement entropy $\mathcal{S} (\vert AB \rangle)$ is a mutual information measure that detects quantum correlation or ``coherence'' between $A$ and $B$.  If $\mathcal{S} (\vert AB \rangle) = 0$, the joint state factors as $\vert AB \rangle = \vert A \rangle \vert B \rangle$ and hence is {\em separable}; otherwise the joint state is {\em entangled}.  Separable states are also called {\em decoherent}; entangled states are {\em coherent}.  Heuristically, a joint state $|AB \rangle$ is separable if the individual states $|A \rangle$ and $|B \rangle$ can be determined by independent measurements.  If $|AB \rangle$ is entangled, interactions with $A$ and $B$, including measurements, are no longer independent; this lack of independence can be detected \cite{bell:64} and is the empirical basis for demonstrating entanglement \cite{aspect:82} as discussed in \S\ref{pred} below.  

Extending Eq. \eqref{entanglement-entropy} to all times and all bipartite decompositions of $U$ gives a representation of the time-dependent entanglement structure of $U$.  As $t \rightarrow \infty$ the unitary dynamics of Eq. \eqref{prop} will drive $U$ toward maximal entanglement, i.e. $\mathcal{S} (\vert AB \rangle) \rightarrow \dim(\mathcal{H}_B )$ for every bipartite decomposition $U = AB$; where, here again, $B$ is taken to be the smaller of the two systems.  At maximum entanglement, the joint-state evolution can be considered a simple rotation as in Eq. \eqref{phase-rotation}.

Given a decomposition $U = AB$, the Hamiltonian can be decomposed as $H_U = H_A + H_B + H_{AB}$, where $H_{AB}$ represents the $A$-$B$ interaction.  Provided $H_{AB}$ is weak compared to the internal interactions $H_A$ and $H_B$ and the time period of interest is short compared to timescale in which $\mathcal{P}_U$ drives the joint system to maximal entanglement, $A$ and $B$ can be considered at least approximately separable, i.e. $\mathcal{S} (\vert AB \rangle) \approx 0$ so $\vert AB \rangle \approx \vert A \rangle \vert B \rangle$.  In this case, $\vert A \rangle$ and $\vert B \rangle$ can be considered individually well-defined, and bases can be chosen for $A$ and $B$ so that:

\begin{equation} \label{ham}
H_{AB} = \beta^k k_B T^k \sum_i^N \alpha^k_i M^k_i,
\end{equation}
\noindent
where $k =~A$ or $B$, the $M^k_i$ are $N$ Hermitian operators with eigenvalues in $\{ -1,1 \}$, the $\alpha^k_i \in [0,1]$ are such that $\sum^N_i \alpha^k_i = 1$, and $\beta^k \geq \ln 2$ is an inverse measure of $k$'s thermodynamic efficiency that depends on the internal dynamics $H_k$ \cite{fg:20, fgm:21, fm:19, fm:20}.  For fixed $k$, the operators $M^k_i$ clearly must commute, i.e. $[M^k_i, M^k_j] = M^k_i M^k_j - M^k_j M^k_i = 0$ for all $i, j$; hence $H_{AB}$ is swap-symmetric under the permutation group $S_N$ for each $k$.  The thermodynamic factor $\beta^k k_B T^k$ in Eq. \eqref{ham} assures compliance with Landauer's Principle, i.e., assures that the per-bit free-energy cost of classical bit erasure is paid on each cycle (see \cite{fg:20, fgm:21} for discussion).  As $U = AB$ is by assumption isolated, conservation of energy requires $\beta^A T^A = \beta^B T^B$.  As discussed in \cite{fgl:21}, Eq. \eqref{ham} can be written in ordinary narrative form as:

\begin{equation*}
\mathrm{Physical ~Interaction ~=~ (Thermodynamics) \cdot (Yes/No ~questions)}.
\end{equation*}
\noindent
This formulation emphasizes what quantum theory is about: the process of obtaining information.  Obtaining information from $B$ requires, in particular, that $A$ acts on $B$ by asking questions.  As Wheeler \cite{wheeler:89} puts it, ``No question?  No Answer!"  All inference in this framework is active inference; Eq. \eqref{ham} does not allow ``passive perception'' to be a physical process.

In contrast to classical theories of information transfer, in which physical tokens encoding bits are transmitted between observers at different locations, Eq. \eqref{ham} involves no assumptions about spacetime, objects, or motions.  It is strictly topological: given separability, it identifies a boundary $\mathscr{B}$ between $A$ and $B$ at which $H_{AB}$ is defined.  This boundary is the ``channel'' via which $A$ and $B$ exchange strings of bits, ordered by the order of the operators $M^k_i$ in Eq. \eqref{ham}.  Each of these bit strings encodes one eigenvalue of $H_{AB}$; as $H_{AB}$ has units of energy, these eigenvalues are measures of the energy $\beta^k k_B T^k$ exchanged between $A$ and $B$ during each cycle of interaction.  Any time variation of $H_{AB}$ is, therefore, time variation of the energy exchanged and can be written:
\begin{equation} \label{ham-time}
H_{AB} (t) = \beta^k (t) k_B T^k \sum_i^N \alpha^k_i (t) M^k_i,
\end{equation}
\noindent
where $\alpha^k_i (t)$ and $\beta^k (t)$ are subject to the conditions on $\alpha^k_i$ and $\beta^k$ given above.  As noted for Eq. \eqref{prop} - \eqref{phase-rotation} above, the time $t$ in Eq. \eqref{ham-time} is not an observable, clock-referenced time.  We will for simplicity consider $H_{AB}$ to be $t$-invariant; this is effectively an adiabatic assumption.  An observer-specific, measurable time will be introduced in \S\ref{clock} below.

Boundaries such as $\mathscr{B}$ that function as information channels are, when embedded in spacetime and constrained by general covariance,\footnote{Effectively, this is a requirement for consistency with both Special and General Relativity.} {\em holographic screens} that limit communication between the regions they separate to the bits that they encode \cite{hooft:83, susskind:95, bousso:02, almheiri:21}.  Interactions between separable systems, i.e. interactions of the form given by Eq. \eqref{ham} or \eqref{ham-time} can, without loss of generality, be regarded as defined at such holographic screens \cite{fgm:21, fm:20, addazi:21}.  The operators $M^k_i$ can, in this case, be regarded as ``preparation'' and ``measurement'' operators that alternately write and read bit values encoded on $\mathscr{B}$.  The $S_N$ swap symmetry of the $M^k_i$ for each $k$ means that the bits can be prepared, and then measured, in any order and hence independently, provided preparation precedes measurement for each bit (such swaps change the ``zero point'' of the energy scale and are undetectable).  The screen $\mathscr{B}$ is, as is any holographic screen, an ancillary construct, not ``part of'' either $A$ or $B$.  It can be physically realized as an ancillary array of noninteracting qubits (quantum bits) with which $A$ and $B$ interact as illustrated in Fig. \ref{qubit-screen-fig}.

\begin{figure}[H]
\centering
\includegraphics[width=13 cm]{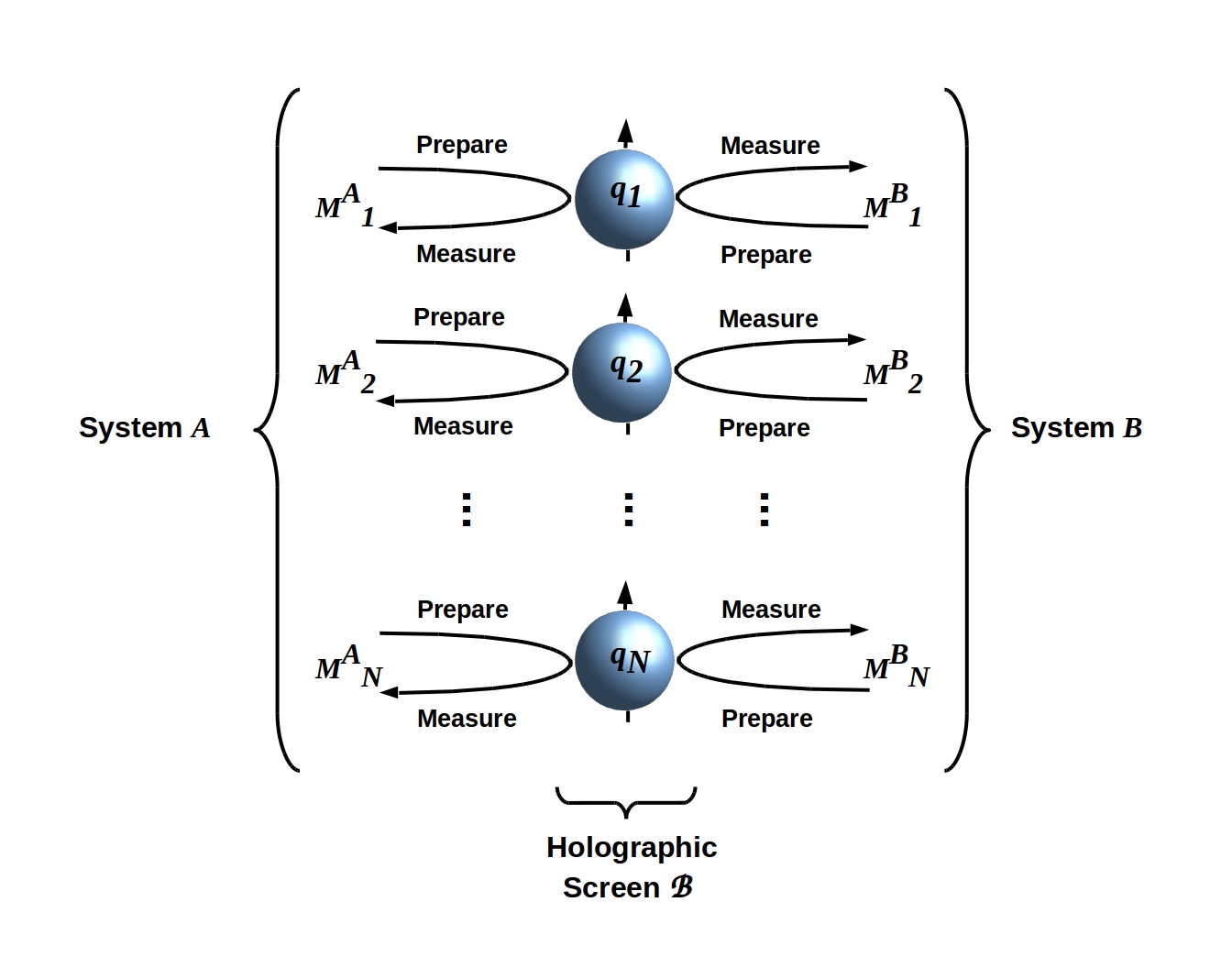}
\caption{A holographic screen $\mathscr{B}$ separating systems $A$ and $B$ with an interaction $H_{AB}$ given by Eq. \eqref{ham} can be realized by an ancillary array of noninteracting qubits that are alternately prepared by $A$ ($B$) and then measured by $B$ ($A$).  Qubits are depicted as Bloch spheres \cite{nielsen:00}.  There is no requirement that $A$ and $B$ share preparation and measurement bases, i.e. QRFs.  Adapted from \cite{fgm:21} Fig. 1, CC-BY license.}
\label{qubit-screen-fig}
\end{figure}

The holographic screen $\mathscr{B}$ has an obvious interpretation in the language of the FEP and active inference: it implements the MB separating $A$ from $B$ \cite{fm:20b}.  ``Active'' and ``sensory'' states of the MB correspond to preparation and measurement implemented by the $M^k_i$; as the interaction $H_{AB}$ is perfectly symmetrical, $A$'s actions are $B$'s sensations and vice-versa.  As bit strings on $\mathscr{B}$ encode energy eigenvalues, what either $A$ or $B$ ``senses'' is energy.  The assumption of separability plays, in this setting, the role played by the assumption of measurability in \cite{friston:19}: it guarantees that systems $A$ and $B$ have well-defined individual states.  These states are at {\em informational} equilibrium; $A$ and $B$ exchange bits one-for-one.  They are not, however, at thermal equilibrium, $T^A = T^B$ unless their thermodynamic efficiencies $\beta^A = \beta^B$.  Hence, when sampled and averaged over macroscopic times, the pure states $\vert A \rangle$ and $\vert B \rangle$ become (mixed) NESS densities $\rho_A$ and $\rho_B$.

\subsection{Reference frames} \label{rf1}

The isolation of $U = AB$ assures that the $A - B$ channel is free of classical, environmentally-induced noise. The preparation and measurement operations $M^A_i$ and $M^B_i$ that $A$ and $B$ use to communicate, however, are defined with respect to and hence depend on the bases $\vert u_i \rangle$ and $\vert v_i \rangle$ of the Hilbert spaces $\mathcal{H}_A$ and $\mathcal{H}_B$ respectively.  In the qubit realization shown in Fig. \ref{qubit-screen-fig}, the $M^A_i$ and $M^B_i$ are $z$-spin operators and so depend on the ``choices'' of $z$ axis $z_A$ and $z_B$.  Provided $A$ and $B$ are separable, and assuming that there are no ``superdeterminist'' {\em a priori} correlations, $z_A$ and $z_B$ are uncorrelated; this is the ``free choice'' assumption.  Free choice is often claimed to be essential to science as a practice \cite{bohr:58, gisin:12}; if it characterizes any bounded system, consistency with quantum theory and special relativity together requires that it must characterize all such systems \cite{conway:09}.  Free choice of a basis introduces quantum ``noise'' that is indistinguishable, observationally, from classical noise.  Suppose $B$ encodes $\vert \uparrow \rangle$ on qubit $q$, using $z_B$ to define the ``up'' direction.  If $z_A = z_B$, $A$'s measured state $\vert q \rangle_A = \vert \uparrow \rangle$, i.e. $A$'s probability $P_A (\vert \uparrow \rangle) = 1$.  If, however, $z_A$ is chosen perpendicular to $z_B$, $\vert q \rangle_A = (\vert \uparrow \rangle - \vert \downarrow \rangle) / \sqrt{2}$ and $P_A (\vert \uparrow \rangle) = P_A (\vert \downarrow \rangle) = 0.5$.

The $z$ axis in Fig. \ref{qubit-screen-fig} is a {\em reference frame}; free choice of the $z$ axis generalizes to free choice of the reference frame for encoding each qubit on $\mathscr{B}$.  While in classical physics, reference frames are typically thought of as fully-specifiable abstractions, in quantum theory reference frames must be considered physical systems -- QRFs -- that encode unmeasurable quantum phase information.  A QRF cannot, therefore, be fully specified by any finite bit string; it is ``nonfungible'' in the terminology of \cite{bartlett:07}.  In the setting of Fig. \ref{qubit-screen-fig}, the $z$ axis deployed by a $A$ ($B$) determines how $A$ ($B$) prepares and then measures the states of the qubits composing $\mathscr{B}$.  We can, in particular, consider $A$ to be isolated from $B$ during the time interval ``between'' preparation and measurement steps, an interval shorter than the natural timescale of the interaction $H_{AB}$.  During such a short interval, we abuse the notation only slightly by writing:

\begin{equation} \label{compute}
\mathcal{P}_k: \vert \mathscr{B} \rangle_{meas} \mapsto \vert \mathscr{B} \rangle_{prep}
\end{equation}
\noindent
where $k = ~A$ or $B$, i.e. by thinking of the internal propagators $\mathcal{P}_A$ and $\mathcal{P}_B$, and hence the dynamics $H_A$ and $H_B$, as computing the next preparation of the qubits on $\mathscr{B}$ from their most recent measured state.  The idea of ``computation'' -- or more properly, of the physical system $A$ ($B$) {\em implementing} a computation -- is the standard notion, reviewed in \cite{horsman:14}: in brief, we can say a system $A$ implements a computation of a function $\mathbf{F}$ if an ``interpretation'' map $\mathcal{I}_{\mathbf{F}}$ exists such that the diagram:

\begin{equation} \label{computation}
\begin{gathered}
\xymatrix@!C=4pc{\langle \mathrm{bits} \rangle_i \ar[r]^{\mathbf{F}} & \langle \mathrm{bits} \rangle_{i+1} \\
\vert \mathscr{B} \rangle_{i} \ar[u]^{\mathcal{I}_{\mathbf{F}}} \ar[r]_{\mathcal{P}_A}  & \vert \mathscr{B} \rangle_{i+1} \ar[u]_{\mathcal{I}_{\mathbf{F}}}}
\end{gathered}
\end{equation}
\noindent
commutes, i.e. $\mathbf{F} ~\mathcal{I}_{\mathbf{F}} = \mathcal{I}_{\mathbf{F}} \mathcal{P}_A$ for every ``step'' $i \rightarrow i+1$, where here $\langle \mathrm{bits} \rangle$ is a finite bit string and $i \rightarrow i+1$ generalizes $meas \rightarrow ~prep$ in Eq. \eqref{compute}.  A QRF, e.g. a $z$ axis as in Fig. \ref{qubit-screen-fig}, implemented by $\mathcal{P}_A$ effectively ``chooses'' the computational basis that renders the string $\langle \mathrm{bits} \rangle$ well-defined. This is the same choice of basis that renders the $M^A_i$ well-defined in Eq. \eqref{ham}.  Hence, we can identify the interpretation map $\mathcal{I}_{\mathbf{F}}$ with the QRF that renders the function $\mathbf{F}$ well-defined; to simplify the notation, we will use boldface $\mathbf{F}$ to label the QRF itself.  The choice of QRF $\mathbf{F}$ and hence of computational basis is {\em functional} or {\em semantic}, not physical; both $H_A$ and $H_{AB}$ are invariant under changes of basis for $\mathcal{H}_A$.  It is worth noting explicitly that the choice of $\mathbf{F}$ is not encoded on $\mathscr{B}$ and is not observationally accessible to $B$.  Indeed, the general question of whether an arbitrary system $A$ implements a QRF $\mathbf{F}$ is Turing undecidable by Rice's theorem \cite{rice:53}; see \cite{fgm:21} for discussion.

As illustrated by the $z$ axes in Fig. \ref{qubit-screen-fig}, the role of any QRF in a measurement setting is to assure that measured values of some degree of freedom can be compared with each other and hence given an operational meaning \cite{fl:20}.  Meter sticks, clocks, gyroscopes, and the Earth with its gravitational and magnetic fields serve this function, and are commonplace laboratory QRFs.  As described in detail in \S\ref{systems} below, recognizing an external object, such as a meter stick, as a QRF and using it as such requires that the observer in question already implements a QRF for the relevant degree of freedom: an observer with no ``internal'' ability to measure and compare lengths, for example, would find a meter stick useless.  All QRFs can, therefore, be considered internal functional components of, or computations implemented by, larger systems that allow measurements of, and assign operational semantics to, one or more degrees of freedom external to the implementing system \cite{fm:19}.  Shared operational semantics across observers requires shared QRFs.  The nonfungibility of QRFs renders the question of QRF sharing in general Turing undecidable, again by Rice's theorem \cite{fgm:21}, a result consistent with the general undecidability of language sharing \cite{quine:60}.

\subsection{Symmetry breaking, decoherence, and agency} \label{breaking}

We will be interested in what follows in quantum systems, considered to be observers, that deploy one or more distinct QRFs to measure particular subsets of the bits encoded on their boundaries/MBs, and that record the values of these bits to a memory that persists for at least one measurement cycle and hence allows comparisons of the values obtained in at least two sequential measurements.  Such systems effectively decompose the holographic screen $\mathscr{B}$ into three disjoint sectors that we will label $E$ the {\em observed environment}, $F$ the {\em unobserved environment}, and $Y$ the {\em memory} sector, respectively.  Equivalently, they effectively decompose the set $M^k_i$ of operators into three disjoint subsets (dropping the redundant index $k$) $M^E_i$, $M^F_j$, and $M^Y_l$, respectively, with:

\begin{equation} \label{sector}
X =_{def} {\rm{dom}}(\{ M^X_i \})
\end{equation}
\noindent
for each sector $X$.  Assuming free choice of basis as above, this decomposition into sectors can be regarded as freely chosen.  Note that as $\mathscr{B}$ is the only locus of classical information in the current formalism, any persistent classical memory {\em must} be a sector on $\mathscr{B}$ as discussed in \cite{fgl:21}.

The assignment of operators to specific sectors breaks the $S_N$ swap symmetry of $\mathscr{B}$ \cite{fgm:21}.  As with choice of QRF, this is a functional or semantic symmetry breaking, not a physical symmetry breaking; the assignment of bits on $\mathscr{B}$ to distinct sectors by $H_A$ has no effect on the definition of $H_{AB}$ given by Eq. \eqref{ham}.  Holding $H_{AB}$ fixed, we can vary the internal interaction $H_A$, and hence the implementation of QRFs as computations, so as to either satisfy $S_N$ symmetry or break it.  Note that varying $H_A$ in this way is equivalent to varying $H_U - H_B$; such variation is undetectable at the boundary $\mathscr{B}$.  This insensitivity of $\mathscr{B}$ to variation in the internal or joint dynamics of $A$ and $B$ is the core physical meaning of the holographic principle \cite{hooft:83, susskind:95, bousso:02, almheiri:21}.

Breaking the $S_N$ swap symmetry on $\mathscr{B}$ renders the sector states $\vert E \rangle$, $\vert F \rangle$, and $\vert Y \rangle$, each of which corresponds to an encoded bit string, separable and hence mutually decoherent \cite{fgm:21}.  As emphasized in \cite{zanardi:01, zanardi:04, dugic:06, dugic:08, torre:10, harshman:11} among others, decoherence is always observer-relative, even when the ``observer'' is an ambient environment.  Hence decoherence due to swap-symmetry breaking is decoherence {\em relative to} $A$; the change-of-basis invariance of $H_{AB}$ renders $A$'s sector boundaries undetectable by $B$.  Given free choice of QRFs, moreover, $B$'s sector boundaries, if any, may be different from $A$'s.  Mismatched sector boundaries between $A$ and $B$ generate apparent ``hidden variables'' and hence variational free energy \cite{friston:10, friston:13, friston:17, friston:19} as discussed in \S\ref{sources} below.  For the present purposes, what is important is that deploying a QRF sensitive to only some degrees of freedom of $\mathscr{B}$ (i.e. sensitive to only some bit values encoded on $\mathscr{B}$) breaks the swap symmetry on $\mathscr{B}$ and creates a decoherent sector as defined by Eq. \eqref{sector}.

Breaking the swap symmetry on $\mathscr{B}$ allows different sectors to have different thermodynamic efficiences, i.e. breaking swap symmetry allows breaking thermodynamic symmetry.  This can be made evident by rewriting Eq. \eqref{ham} for $A$ in terms of sectors $X$ as:
\begin{equation} \label{ham-sector}
H_{AB} = \sum_X \beta^X k_B T^A \sum_{i \in i_X } \alpha^A_i M^A_i,
\end{equation}
\noindent
where $i_X$ is the set of indices of operators $M^A_i$ assigned to sector $X$.  Note that this symmetry breaking does not change the total energy exchanged (i.e. the eigenvalue of $H_{AB}$), it only allocates the energy differently to different sectors.  This allows sectors to perform different thermodynamic roles; in particular, it allows the unobserved environment sector $F$ to serve as a source of free energy to -- and a sink of waste heat from -- the observed environment sector $E$ and the memory sector $Y$ \cite{fg:20, fgm:21}.  The bits encoded on $F$ are, therefore, noninformative to $A$ and are traced (effectively, marginalised and averaged) over when defining information-bearing states of $A$.  Mathematically, this trace operation can be viewed as implementing decoherence as discussed in \cite{omnes:92, zurek:98, zurek:03, schloss:03, schloss:07}, i.e., as implementing the sector boundary of $F$.  Assigning this thermodynamic function to $F$ enables thermodynamically-expensive classical information processing of bits encoded on $E$ and $Y$, including maintaining the stability of bit values written to $Y$ as discussed in \S\ref{clock} below.  It therefore allows $E$ and $Y$ to have distinct semantics. 

Free choice of a decomposition of $\mathscr{B}$ into sectors with different QRF-induced semantics is indicative of agency.  Hence, we can define:

\begin{definition} \label{agent-def}
A (nontrivial) {\em agent} is a system $A$ with an internal dynamics $H_A$ that breaks the $S_N$ swap symmetry of its boundary/MB $\mathscr{B}$.
\end{definition}
\noindent
We can consider a system with a swap-symmetric boundary as shown in Fig. \ref{qubit-screen-fig} a {\em trivial} agent.  Trivial agents correspond to ``inert'' systems with functionally-insignificant internal states, and hence no internal information-processing capacity, as discussed in \cite{friston:19}.  All nontrivial agents are ``cognitive'' systems that engage in active inference.

Before proceeding to further characterize these sectors or to consider the imposition of additional structure on $E$ by further QRFs, we briefly review below the specification of QRFs using the generic formalism of Channel Theory \cite{barwise:97}.  This formalism will enable characterization of the asymptotic behavior of the FEP and of context-switching as a mechanism for minimizing FEP as discussed below in \S\ref{asymptotic} below.

\subsection{Channel theory of QRFs} \label{chan-th}

Channel Theory \cite{barwise:97} is an application of the category theory of Chu spaces, spaces of semantic relations exemplified by object -- attribute tables \cite{barr:79, pratt:99a, pratt:99b}.  Indeed Channel Theory can be regarded as defining a category $\mathbf{Chan}$ that is isomorphic to the category $\mathbf{Chu}(\mathbf{Set},K)$ (for short, $\mathbf{Chu}$) of Chu spaces.   While conceptually simple, Channel Theory is surprisingly rich, providing both a natural representation of conditional probabilities and formal criteria for operator commutativity and quantum contextuality \cite{fg:19a, fg:21} (for a logico-philosophical perspective on the properties of semantic coherence in $\mathbf{Chan}$, see e.g. \cite{collier:11}). Furthermore, it provides an implementation-independent formal language for writing functional specifications of QRFs as computations.

The central idea of Channel Theory is that of a ``classifier'' that relates tokens in some language to types in that language.  We can define a classifier as an object in $\mathbf{Chan}$ as follows:

\begin{definition}\label{class-def-1}
A {\em classifier} $\mathcal{A}$ is a triple $\langle Tok(\mathcal{A}), Typ(\mathcal{A}), \models_{\mathcal{A}} \rangle$ where $Tok(\mathcal{A})$ is a set of ``tokens'', $Typ(\mathcal{A})$ is a set of ``types'', and $\models_{\mathcal{A}}$ is a ``classification'' relating tokens to types.
\end{definition}
\noindent
The classification $\models_{\mathcal{A}}$ can, in general, be valued in any set $K$ without assumed structure (as is the case for Chu spaces); for simplicity, we will consider only binary classifications.  Morphisms in $\mathbf{Chan}$, called ``infomorphisms'' between these objects are then given by the following:

\begin{definition}\label{class-def-2}
Given two classifiers $\mathcal{A} = \langle Tok(\mathcal{A}), Typ(\mathcal{A}), \models_{\mathcal{A}} \rangle$ and $\mathcal{B} = \langle Tok(\mathcal{B}), Typ(\mathcal{B}), \models_{\mathcal{B}} \rangle$, an {\em infomorphism} $f: \mathcal{A} \rightarrow \mathcal{B}$ is a pair of maps $\overrightarrow{f}: Tok(\mathcal{B}) \rightarrow Tok(\mathcal{A})$ and $\overleftarrow{f}: Typ(\mathcal{A}) \rightarrow Typ(\mathcal{B})$ such that $\forall b \in Tok(\mathcal{B})$ and $\forall a \in Typ(\mathcal{A})$, $\overrightarrow{f}(b) \models_{\mathcal{A}} a$ if and only if $b \models_{\mathcal{B}} \overleftarrow{f}(a)$.
\end{definition}
\noindent
This last definition can be represented schematically as the requirement that the following diagram commutes:

\begin{equation}\label{info-diagram-1}
\begin{gathered}
\xymatrix@!C=3pc{\rm{Typ}(\mathcal{A}) \ar[r]^{\overrightarrow{f}}   & \rm{Typ}(\mathcal{B}) \ar@{-}[d]^{\models_{\mathcal{B}}} \\
\rm{Tok}(\mathcal{A}) \ar@{-}[u]^{\models_{\mathcal{A}}}  & \rm{Tok}(\mathcal{B}) \ar[l]_{\overleftarrow{f}}}
\end{gathered}
\end{equation}
\noindent
An infomorphism $f: \mathcal{A} \rightarrow \mathcal{B}$ is, effectively, a map relating the semantic constraints imposed by the classification $\models_{\mathcal{A}}$ to those imposed by $\models_{\mathcal{B}}$.

We are, in practice, interested in collections of infomorphisms than construct complex semantic constraints out of simple ones; these will allow us to specify QRFs as hierarchies of semantic constraints.  Given a finite collection $\mathcal{A}_i$ of classifiers, we can represent this construction process as a finite, commuting {\em cocone diagram} (CCD) depicting a flow of infomorphisms sending inputs to a {\em core} $\mathbf{C^\prime}$ that is the category-theoretic colimit of the underlying classifiers, i.e., is the apex of the maximally general diagram of this form over the  $\mathcal{A}_i$, if a unique such maximum exists:

\begin{equation}\label{ccd-1}
\begin{gathered}
\xymatrix@C=4pc{&\mathbf{C^\prime} &  \\
\mathcal{A}_1 \ar[ur]^{f_1} \ar[r]_{g_{12}} & \mathcal{A}_2 \ar[u]_{f_2} \ar[r]_{g_{23}} & \ldots ~\mathcal{A}_k \ar[ul]_{f_k}
}
\end{gathered}
\end{equation}
\noindent
The cocone core $\mathbf{C^\prime}$ is itself a classifier that encodes, via the incoming infomorphisms $f_i$, the conjunction of the semantic constraints imposed by the $\mathcal{A}_i$.

There is a dual construction to this CCD, namely a commuting finite \emph{cone diagram} (CD) of infomorphisms on the same classifiers, where all arrows are reversed.  In this case the core of the (dual) channel is the category-theoretic limit of all possible downward-going structure-preserving maps to the classifiers $\mathcal{A}_i$.  Hence, we can define the central idea of a finite, commuting {\em cone-cocone diagram} (CCCD) as consisting of both a cone and a cocone on a single finite set of classifiers $\mathcal{A}_i$ linked by infomorphisms as depicted below:

\begin{equation}\label{cccd-1}
\begin{gathered}
\xymatrix@C=6pc{&\mathbf{C^\prime} &  \\
\mathcal{A}_1 \ar[ur]^{f_1} \ar[r]_{g_{12}}^{g_{21}} & \ar[l] \mathcal{A}_2 \ar[u]_{f_2} \ar[r]_{g_{23}}^{g_{32}} & \ar[l] \ldots ~\mathcal{A}_k \ar[ul]_{f_k} \\
&\mathbf{D^\prime} \ar[ul]^{h_1} \ar[u]^{h_2} \ar[ur]_{h_k}&
}
\end{gathered}
\end{equation}
\noindent
If the cores $\mathbf{C^\prime} = \mathbf{D^\prime}$, we can also represent the CCCD as:

\begin{equation}\label{cccd-2}
\begin{gathered}
\xymatrix@C=6pc{\mathcal{A}_1 \ar[r]_{g_{12}}^{g_{21}} & \ar[l] \mathcal{A}_2 \ar[r]_{g_{23}}^{g_{32}} & \ar[l] \ldots ~\mathcal{A}_k \\
&\mathbf{C^\prime} \ar[ul]^{h_1} \ar[u]^{h_2} \ar[ur]_{h_k}& \\
\mathcal{A}_1 \ar[ur]^{f_1} \ar[r]_{g_{12}}^{g_{21}} & \ar[l] \mathcal{A}_2 \ar[u]_{f_2} \ar[r]_{g_{23}}^{g_{32}} & \ar[l] \ldots ~\mathcal{A}_k \ar[ul]_{f_k}
}
\end{gathered}
\end{equation}
\noindent
This diagram is naturally interpreted as reconstructing the semantics of the $\mathcal{A}_i$ via the ``combined'' representation $\mathbf{C^\prime}$.  Generalizing Diagram \eqref{cccd-2} by letting $\mathbf{C^\prime}$ be the limit of a smaller set of classifiers $\mathcal{A}^{\prime}_1, \dots \mathcal{A}^{\prime}_j, ~j < k$, we can write:

\begin{equation}\label{cccd-3}
\begin{gathered}
\xymatrix@C=6pc{\mathcal{A}^{\prime}_1 \ar[r]_{g^{\prime}_{12}}^{g^{\prime}_{21}} & \ar[l] \mathcal{A}^{\prime}_2 \ar[r]_{g^{\prime}_{23}}^{g^{\prime}_{32}} & \ar[l] \ldots ~\mathcal{A}^{\prime}_j \\
&\mathbf{C^\prime} \ar[ul]^{h^{\prime}_1} \ar[u]^{h^{\prime}_2} \ar[ur]_{h^{\prime}_j}& \\
\mathcal{A}_1 \ar[ur]^{f_1} \ar[r]_{g_{12}}^{g_{21}} & \ar[l] \mathcal{A}_2 \ar[u]_{f_2} \ar[r]_{g_{23}}^{g_{32}} & \ar[l] \ldots ~\mathcal{A}_k \ar[ul]_{f_k}
}
\end{gathered}
\end{equation}
\noindent
Diagram \eqref{cccd-3} provides a natural representation of coarse-graining the semantics of $\mathcal{A}_i$ via $\mathbf{C^\prime}$ into a compressed representation $\mathcal{A}^{\prime}_i$.  We will employ this generalization in \S\ref{clock} below to specify the writing of coarse-grained records of observational outcomes to the memory sector $Y$.

Diagrams such as \eqref{ccd-1} -- \eqref{cccd-3} can be generalized into hierarchical networks by adding intermediate layers of classifiers and appropriate maps; when this is done, they clearly resemble artificial neural networks (ANNs), and in the ``bowtie'' form of Diagrams \eqref{cccd-2} and \eqref{cccd-3}, variational autoencoders (VAEs) \cite{fg:19a}.  The core $\mathbf{C^\prime}$ in \eqref{cccd-2} and \eqref{cccd-3} can be viewed as both an ``answer'' computed by the $f_i$ from inputs to the $\mathcal{A}_i$ and, dually, as an ``instruction'' propagated by the $h_i$ (or in Diagram \eqref{cccd-3}, the $h^{\prime}_i$), to drive outputs from the $\mathcal{A}_i$ (or in Diagram \eqref{cccd-3}, the $\mathcal{A}^{\prime}_i$).  Such dual input/output behavior is exactly the behavior of a QRF.  We can, therefore, represent any QRF as a CCCD ``attached'' to a subset of measurement operators $M^A_k, \dots M^A_n$ by maps that identify the binary eigenvalues of the $M^A_i$ with binary inputs to the $\mathcal{A}_i$ as illustrated in Fig. \ref{CCCD-to-screen-fig}.

\begin{figure}[H]
\centering
\includegraphics[width=15 cm]{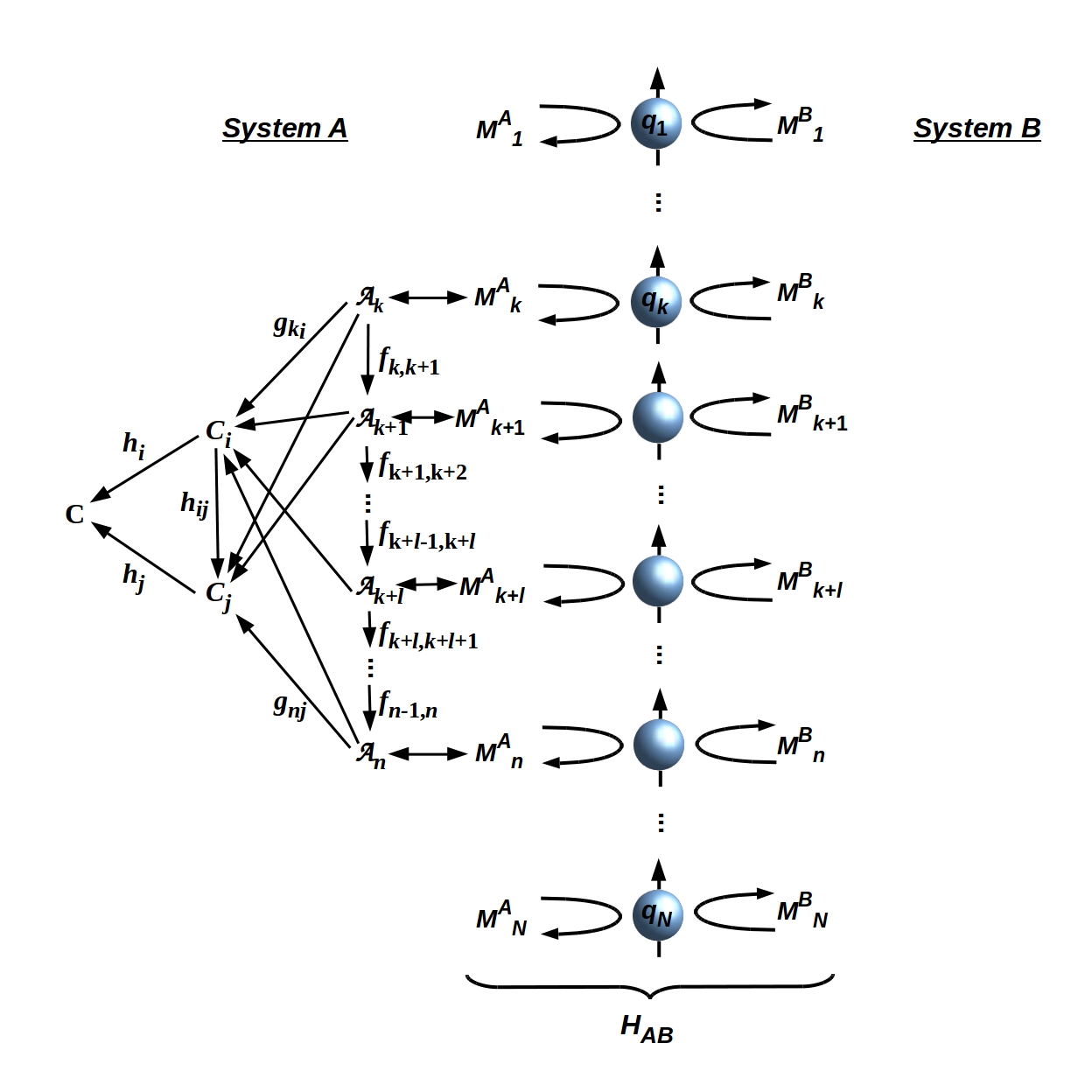}
\caption{Attaching a CCCD to a subset of measurement operators $M^A_k, \dots M^A_n$ by identifying the binary eigenvalues of the $M^A_i$ with binary inputs to the $\mathcal{A}_i$.  Only the CCD direction arrows are shown for simplicity; adding equivalent but reversed arrows completes the CCCD.  The CCCD specifies a function computed by the internal dynamics $\mathcal{P}_A$, i.e. a QRF deployed by $A$.  Adapted from \cite{fgm:21} Fig. 3; CC-BY license.}
\label{CCCD-to-screen-fig}
\end{figure}

Thus far we have considered binary CCCDs, corresponding to Boolean ANNs or VAEs or to QRFs imposing Boolean constraints.  Mapping a CCCD to an arbitrary ANN or VAE that computes some function of interest $\mathbf{F}$ just requires assigning probabilities, i.e. ``weights'' to the infomorphisms in a way that respects the Kolmogorov axioms \cite{fg:19a, fg:21}.  We can then treat the computed function $\mathbf{F}$ from network inputs to network outputs as an arbitrary probabilistic QRF, or as suggested by Diagram \eqref{cccd-3}, a pair of coupled QRFs that process some ``sensory'' signal received by the $\mathcal{A}_i$ and write the result to a lower-dimensional and hence coarse-grained ``memory'' via the $\mathcal{A}^{\prime}_i$.  A QRF so defined is naturally interpreted as performing hierarchical Bayesian inference \cite{fg:19a,fg:21}; see \cite{fg:19b,fg:20b} for applications to human cognition.

To assign probabilities to infomorphisms, it is convenient to add the structure of a ``local logic'' $\mathcal{L}(\mathcal{A})$ relating subsets of tokens and types to each classifier $\mathcal{A}$ \cite{barwise:97}.  To do this, we define an ``implication'' relation between subsets of types:
\begin{definition}\label{prob-1}
Two subsets $M, N \subseteq Typ(\mathcal{A})$ are related by a {\em sequent} $M \models_{\mathcal{A}} N$ if $\forall x \in Tok(\mathcal{A}), x \models_{\mathcal{A}} M \Rightarrow x \models_{\mathcal{A}} N$.
\end{definition}
\noindent
Via the sequent relation, the local logic $\mathcal{L}(\mathcal{A})$ effectively arranges the types of $\mathcal{A}$ into a hierarchy; any classifier can be extended to a local logic in this way by adding types as needed.  Considering all classifiers to be extended to local logics, Diagrams \eqref{ccd-1} -- \eqref{cccd-3} can be considered diagrams of local logics by requiring each of the infomorphisms to be a ``logic infomorphism'' that preserves the sequent structure.  In general, given an information flow channel:
\begin{equation}\label{channel-flow-1}
\longrightarrow \mathcal{A}_{\alpha -1} \longrightarrow \mathcal{A}_{\alpha} \longrightarrow \mathcal{A}_{\alpha +1} \longrightarrow \cdots
\end{equation}
\noindent
the semantic content can be extended by postulating local logics $\mathcal{L}_{\alpha} = \mathcal{L}(\mathcal{A}_{\alpha})$ generated by the corresponding classifiers $\mathcal{A}_{\alpha}$ (assumed, in principle, to be in relationship to a (regular) theory associated to the individual $\mathcal{A}_{\alpha}$, as specified in \cite[Ch 9]{barwise:97} (cf. \cite{barwise:97a}) and reviewed in \cite[Appendix A]{fg:21}), so to obtain a flow of logic infomorphisms:
\begin{equation}\label{channel-flow-2}
\cdots \longrightarrow \mathcal{L}_{\alpha -1} \longrightarrow \mathcal{L}_{\alpha} \longrightarrow \mathcal{L}_{\alpha +1} \longrightarrow \cdots
\end{equation}
\noindent
which we can take to comprise comprise a CCD as in Eq. \eqref{ccd-1}.  Logic infomorphisms are, effectively, embeddings of type hierarchies; Diagrams \eqref{cccd-2} and \eqref{cccd-3} can be viewed as embeddings into a ``top-level'' type hierarchy $\mathbf{C^\prime}$ that assigns an overall semantics to its inputs, followed by encodings of this top-level hierarchy into some componential representation.

The local logic $\mathcal{L}(\mathcal{A})$ defined above is Boolean.  To extend the type hierarchy defined by $\mathcal{L}(\mathcal{A})$ to a hierarchical Bayesian inference, we extend $\mathcal{L}(\mathcal{A})$ to a probabilistic logic by relaxing the sequent relation to require only that if $x \models_{\mathcal{A}} M$, there is some probability $P(N|M)$ that $x \models_{\mathcal{A}} N$.  We can then write:
\begin{equation}\label{prob-2}
M \models_{\mathcal{A}}^{P} N =_{def} P( M \vert N)
\end{equation}
\noindent
and construct (extended, probabilistic) logic infomorphisms as above, requiring that they preserve the conditional probabilities $P( M \vert N)$ for all subsets $M, ~N$ at each level of the hierarchy.\footnote{As pointed out in \cite{allwein:04}, it is instructive to see that Eq. \eqref{prob-2} reveals how a conditional probability can be used for interpreting the logical implication ``$\Rightarrow$''
as discussed in \cite{adams:98}.}
Traversing upwards in the hierarchy then imposes multiple conditioning on each ``low-level'' probability distribution; traversing downwards sequentially unpacks this conditioning.  In fundamental Bayesian terms, $M$ above can be regarded as a previous event, whether observed or conjectured, and $N$ as a currently observed datum, in which case $P(M)$ becomes the prior, and $P(N)$ the evidence, that together generate a prediction. Given the likelihood $P(N \vert M)$ as the conditional obtained from weakening the sequent via Eq. \eqref{prob-2}, Bayes' theorem specifies this conditional as the posterior:
\begin{equation}\label{prob-3}
P(M \vert N) = \frac{P(N\vert M) P(M)}{P(N)}.
\end{equation}
A brief example illustrates these principles as follows: consider arbitrary classifiers $\mathcal{A}_1^{(a)}, \ldots,  \mathcal{A}_5^{(e)}$ in some part of an information channel where (as in \cite{fg:21}) the classifiers correspond to events $a,b,c,d,e$, respectively, together with logic infomorphisms $f_{13}, \ldots, f_{45}$ between them, in which the sequents are relaxed to conditional probabilities via Eq. \eqref{prob-2}:
\begin{equation}
\begin{gathered}
\xymatrix{&\mathcal{A}_1^{(a)} \ar[dl]_{f_{13}} \ar[dr]^{f_{14}} &  & \mathcal{A}_2^{(b)}\ar[dl]_{f_{24}} \\
\mathcal{A}_3^{(c)} & & \mathcal{A}_4^{(d)} \ar[d]^{f_{45}}  &  & \\
& & \mathcal{A}_5^{(e)}
}
\end{gathered}
\end{equation}
\noindent
Following e.g. \cite{cherniak:91}, this particular channel then generates a joint probability distribution given by:
\begin{equation}
p(abcde) = p(a)p(b)p(c\vert a)p(d \vert ab)p(e\vert d).
\end{equation}
Putting these details within the framework of the above diagrams, a portion of a typical CCD computing a hierarchical Bayesian inference from a set of (posterior) observations $\mathcal{A}_i$ to an outcome $\mathbf{C}'$ has the form:
\begin{equation}\label{ccd-prob-1}
\begin{gathered}
\xymatrix@C=4pc{&\mathbf{C^\prime} &  \\
\mathcal{A}_1 \ar[ur]^{p_{10}(\cdot \vert \cdot)} \ar[r]_{p_{12}(\cdot \vert \cdot)}& \mathcal{A}_2 \ar[u]_{p_{20}(\cdot \vert \cdot)} \ar[r]_{p_{23}(\cdot \vert \cdot)} & \ldots ~\mathcal{A}_k \ar[ul]_{p_{k0}(\cdot \vert \cdot)}
}
\end{gathered}
\end{equation}

In this formal setting, the diagram commutativity of Diagrams \eqref{ccd-1} -- \eqref{cccd-3}, with probabilities added as in Diagram \eqref{ccd-prob-1} above, enforces Bayesian coherence: the probability associated with any arrow $\mathcal{U} \rightarrow \mathcal{V}$ must equal the product of the probabilities associated with the arrows on any other directed path from $\mathcal{U}$ to $\mathcal{V}$.\footnote{Probability spaces and conditional distributions can also be defined directly for Chu spaces as exhibited in \cite[Exs. 2.4,~2.5]{fg:21}.  See \cite{smithe:21} for an alternative representation of Bayesian inference and derivation of the FEP using monodial categories and an operational formalism closely related to those employed in categorical quantum theory \cite{coecke:10}.}  Any pair of subsets of the $\mathcal{A}_i$ have, therefore, a well-defined joint probability distribution.  As shown in \cite{fg:21}, the converse is also true.  Letting $\mathcal{A}_1, \mathcal{A}_2 \dots \mathcal{B}$ and $\mathcal{B}, \mathcal{C}_1, \dots \mathcal{C}_k$ be finite sets of classifiers, we can state the following \cite[Thm. 7.1]{fg:21}:

\begin{theorem} \label{thm1}
A well-defined joint probability distribution exists over subsets $\mathcal{A}_1, \mathcal{A}_2 \dots \mathcal{B}$ and $\mathcal{B}, \mathcal{C}_1, \dots \mathcal{C}_k$ of classifiers if and only if the following diagram commutes:
\begin{equation} \label{double-ccd}
\begin{gathered}
\xymatrix@C=4pc{& &\mathbf{C} & & \\
& \mathbf{C_1} \ar[ur]^{\phi} & & \mathbf{C_2} \ar[ul]_{\psi} & \\
\mathcal{A}_1 \ar[ur]^{f_1} \ar[r] & \mathcal{A}_2 \ar[u]_{f_2} \ar[r] & \ldots ~\mathcal{B} \ar[ul]_{f_B} \ar[ur]^{g_B} \ar[r] & \mathcal{C}_1 \ar[u]_{g_2} \ar[r] & \ldots ~\mathcal{C}_k \ar[ul]_{g_k}
}
\end{gathered}
\end{equation}
\noindent
that is, if and only if there exist logic infomorphisms $\phi$, $\psi$ and a classifier $\mathbf{C}$ such that the above diagram is a CCD.
\end{theorem}

\begin{proof}
See \cite{fg:21}.  The key observation is that $\phi$ and/or $\psi$ can only fail to exist if the joint probability distribution fails to exist.
\end{proof}
\noindent
The use of overlapping subsets of classifiers in Theorem \ref{thm1} mirrors the use of mutually-noncommuting subsets of mutually-commuting observables in the canonical definition of quantum contextuality \cite{kochen:67, mermin:93} and its extension to a purely statistical criterion of incompatibility between measurement contexts \cite{dzha:17b, dzh:18}.  How the deployment of overlapping but mutually-noncommuting subsets of QRFs implements context-switching between pairs of complementary observables -- such as position and momentum -- is discussed in \S\ref{noncomm} below.  How the FEP can drive context-switching as a component of active inference is then discussed in \S\ref{asymptotic} below.

Viewing commutativity -- and hence Bayesian coherence -- as fundamental to the definition of measurement, and hence also to the definition of preparation or action on the environment generally, suggests that the Born rule, i.e. that if a state:

\begin{equation}
\vert X \rangle = \sum_i \alpha_i \vert x_i \rangle,
\end{equation}
\noindent
the probability of obtaining $\vert x_i \rangle$ as an observational outcome is:

\begin{equation}
P(\vert x_i \rangle = \vert \alpha_i \vert^2,
\end{equation}
\noindent
is a prescription for coherently assigning amplitudes $\alpha_i$ to components $\vert x_i \rangle$, consistent with the position advocated in \cite{fuchs:13, mermin:18}.


\section{Repeated measurements and system identification} \label{3}

\subsection{Memory, time, and coarse-graining} \label{clock}

The idea that a system must possess a (quasi-) NESS solution to its density dynamics -- and hence be restricted to trajectories in its classical configuration space that do not diverge exponentially over time -- in order to be observable as a ``thing'', immediately raises the issues of time as measurable duration and of memory as persistent over measurable time.  As the simplest case, consider the NESS density $\rho_E$ of the observed environment sector $E$ of some agent $A$.  The agent $A$ can only detect changes in $\rho_E$ and employ them as bases for inferences and actions if $A$ can write time-ordered records $[\rho_E (t_A)]$ to, and subsequently read them from, the memory sector $Y$.  Here $t_A$ is an $A$-specific time coordinate that must be constructed.  Persistence of the memory record $[\rho_E (t_A)]$ between writing and reading requires that it be ``protected'' from environmental noise, i.e. from ongoing events in $E$; hence if $\vert Y \rangle$ is the state encoding $[\rho_E (t_A)]$, $|Y \rangle$ must vary only slowly under the action of $H_B$.  In the notation of Eq. \eqref{ham}, ``protection'' occurs if the coefficients $\alpha^B_i$ of operators $M^B_i$ acting on qubits within $Y$ all have small magnitudes.  We can, therefore, distinguish two classes of actions by $A$ on its boundary/MB $\mathscr{B}$.  Actions on $E$ are ``questions'' in Wheeler's \cite{wheeler:89} sense; they provoke informative responses from $B$ that can be used to develop a generative model of $H_B$ as discussed below.  Actions on $Y$ are recordings that must remain relatively stable to be useful.  Stability of the memory sector $Y$ against $H_B$ is a critical resource for any system $A$ capable of responding to environmental state changes, and hence of any $A$ capable of active inference. Heuristically, this is clearly evinced in active vision, where we actively palpate the world, every 250 ms or so, to update our working memory $Y$; i.e., update our (Bayesian) beliefs about the observed environment $E$.

The Second Law tells us that protecting any state against noise requires the expenditure of free energy.  Hence, any agent $A$ is faced with a fundamental thermodynamic tradeoff: maintaining the stability of $Y$ requires free energy sourced from the unobserved (traced over) environmental sector $F$.  For fixed $H_{AB}$, and hence fixed $\mathscr{B}$, expanding $F$ to obtain more free energy for $Y$ requires shrinking $E$, i.e. ``seeing'' less of the environment.  The alternative is to coarse-grain $Y$ either in time or in space, i.e. in bit-string length.  Which horn of this thermodynamic trilemma a given system takes is determined by the QRFs it implements. Classical manifestation of this trade-off are seen in many guises. For example, the intimate relationship between the efficiency of information transfer afforded by compression -- as seen in minimum message length formulations of variational free energy \cite{MacKay:95,Wallace:99} -- through to the minimisation of statistical complexity afforded by coarse-graining and quantisation \cite{Smith:19}.

The simplest memory-write process is illustrated in cartoon form in Fig. \ref{mem-write-fig}.  A state $\vert E \rangle$, a bit string of length $\dim(E)$, is constructed from one-bit operators $M^E_i$ by a QRF $\mathbf{E}$.  It is then written to the memory $Y$ by a QRF $\mathbf{Y}$, with free energy sourced from the remainder of $\mathscr{B}$, i.e. from the unobserved sector $F$.  In the simplest case, the memory capacity $\dim(Y) = n \dim(E) + \log_2 n$ where $n$ is the number of distinguishable records.  The $\log_2 n$ labels that allow records to be distinguished can, without loss of generality, be considered to be an integer sequence of clock ticks $i \rightarrow i+1$, starting from $i = 1$.  Hence, any memory with more than one-record capacity defines a clock $\mathcal{G}_{ij}$, which we will see below must, in general, be a groupoid operator \cite{fg:20}.\footnote{Recall that a {\em groupoid} is a category in which the objects and arrows each comprise a set (as for a `small' category), and every arrow is invertible \cite{weinstein:96, brown:06}. A groupoid generalizes the {\em group} concept in so far that the former can admit ``multiple identities'' (the objects).}  This clock is an internal time QRF that defines the time coordinate $t_A$.

\begin{figure}[H]
\centering
\includegraphics[width=13 cm]{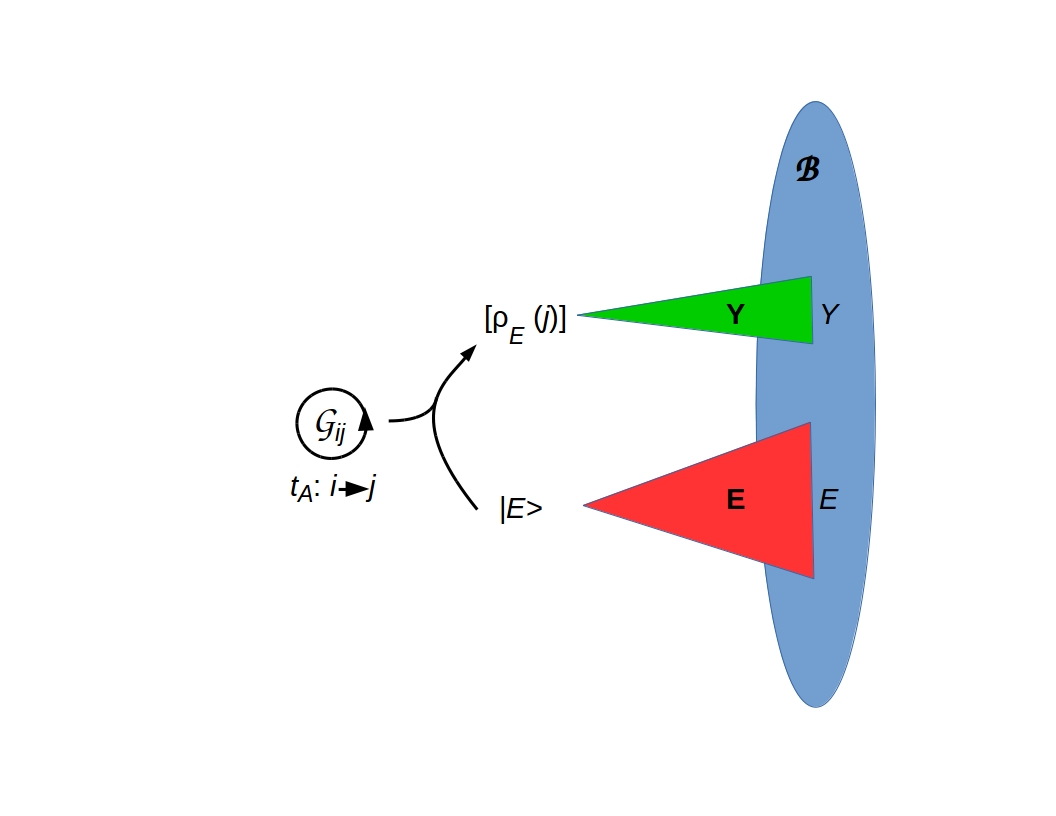}
\caption{Cartoon illustration of QRFs required to write a readable memory of an observed environmental state $\vert E \rangle$.  The QRFs $\mathbf{E}$ and $\mathbf{Y}$ read the state from $E$ and write it to $Y$ respectively.  The clock $\mathcal{G}_{ij}$ is a time QRF that defines the time coordinate $t_A$.}
\label{mem-write-fig}
\end{figure}

The thermodynamic trade-off faced by $A$ is between $\dim(E)$, $\dim(Y)$, and the sampling time of states $\vert E \rangle$, which determines the length of one clock tick and hence of one unit of $t_A$.  As mentioned earlier, the minimum timescale for biological systems is set by the thermal dissipation time to roughly 50 fs.  Practical biological clocks run much more slowly; a clock based on Gamma-band neural activity in mammalian cortex, for example, has a sampling time of 10 -- 20 ms.  The recorded record, in this case, is a sample from or average over a time ensemble $< \vert E \rangle >$ of measured states, i.e., is a record $[ \rho_E ]$ of a coarse-grained density.  This record can be further compressed to achieve $\dim([ \rho_E ]) < \dim(E)$ and hence $\dim(Y) < n \cdot \dim(E) + \log_2 n$.

With these QRFs, the memory read-compare-write cycle can be represented as in Fig. \ref{write-read-fig}.  The classical record $[ \rho_E(i) ]$ written in the previous cycle is read by $\mathbf{Y}$ at ``external'' (i.e. Eq. \eqref{prop} or \eqref{ham-time}) time $t$.  It is compared to the current measured state $\vert E(t) \rangle$ and a new record $[ \rho_E(j) ]$ is written by $\mathbf{Y}$ at $t + \Delta t$.  This read-compare-write cycle advances the internal clock $\mathcal{G}_{ij}$ by one tick $i \rightarrow j$.  All QRFs together with the comparison function are implemented by the internal dynamics $\mathcal{P}_A$.  Formally, we can think of $\mathcal{P}_A$ as a weighted sum of ``all possible paths'' from the boundary state $\vert \mathscr{B} (t) \rangle$ to the boundary state $\vert \mathscr{B} (t + \Delta t) \rangle$ as in Eq. \eqref{compute} \cite{deutsch:02}; see \cite{marciano:21} for discussion.

\begin{figure}[H]
\centering
\includegraphics[width=15 cm]{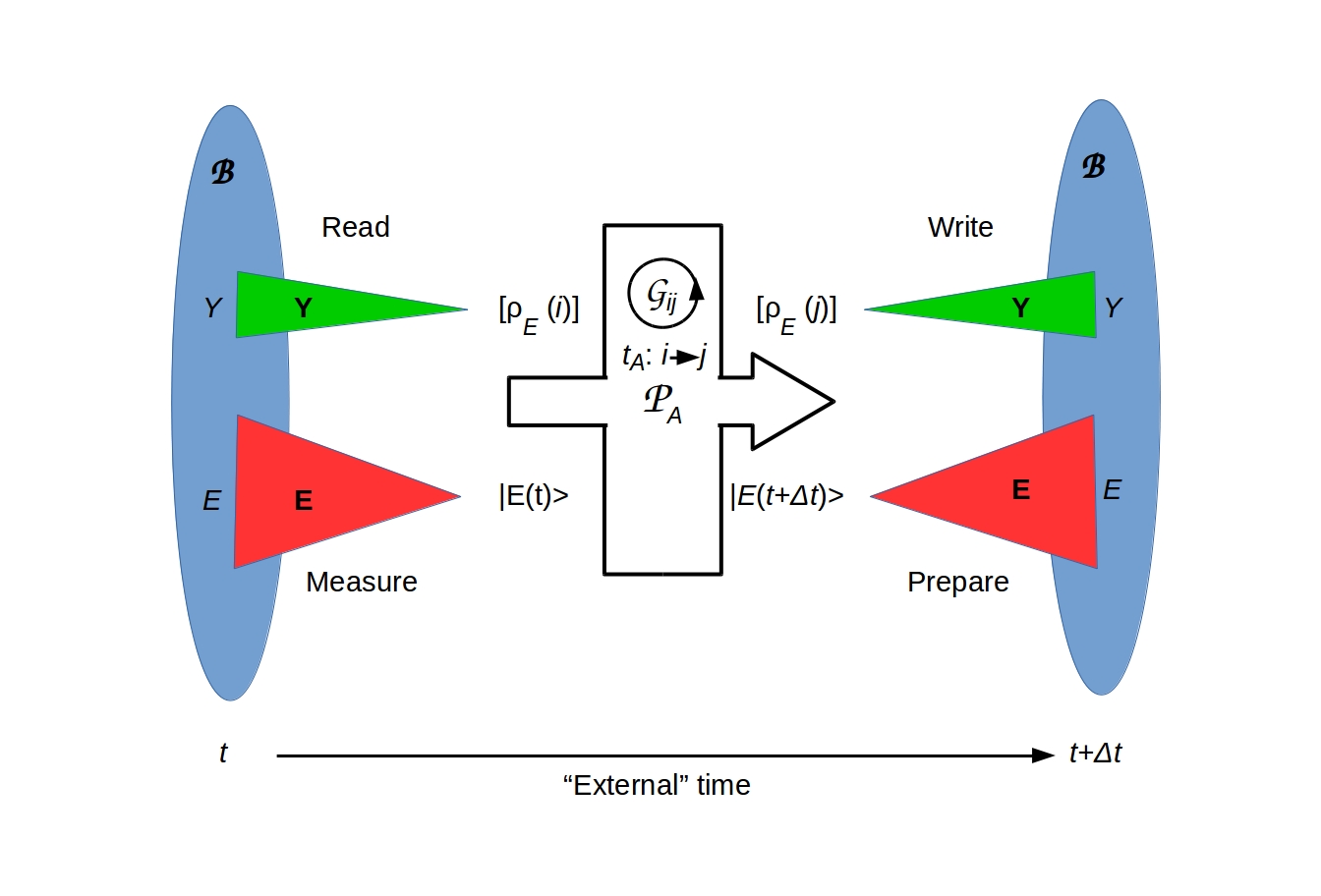}
\caption{Cartoon illustration of one memory read-compare-write cycle defining one tick $i \rightarrow j$ of the internal clock $\mathcal{G}_{ij}$ and requiring an interval $\Delta t$ of ``external'' time $t$.  All QRFs and the comparison function are implemented by the quantum dynamics $\mathcal{P}_A$ (block arrow).}
\label{write-read-fig}
\end{figure}

Fig. \ref{write-read-fig} explicitly illustrates an important distinction between classical and quantum representations of dynamics and hence between classical and quantum formulations of the FEP.  Classical physics assumes a spacetime embedding; hence the MB of any topologically-connected system can be associated with a spatial boundary that follows a smooth spacetime trajectory.  ``Internal'' states of the system and the ``internal'' dynamics that operates on those states to maintain a NESS are spatially localized inside this boundary, and are, in particular, not exposed {\em on} the boundary.  The current, quantum setting makes, in contrast, no assumptions about spacetime embedding as emphasized earlier.  The boundary $\mathscr{B}$ represents a Hilbert-space decomposition, not a ``physical'' spacetime decomposition.  As $\mathscr{B}$ is the only locus of classical information, ``internal'' classical states, including all classical memory records, are exposed on $\mathscr{B}$.  The assumption that these exposed states are ``protected'' from $H_B$ corresponds to the classical assumption of ``self-evidencing'': that maintaining the integrity of the MB as a spatial boundary enables the maintenance of the NESS and vice-versa \cite{friston:19}.

The true internal state of a quantum system $A$ is the state $\vert A \rangle$ that remains independently well-defined provided the joint state $\vert AB \rangle$ remains separable, i.e. provided the interaction $H_{AB}$ can be written as Eq. \eqref{ham}.  This internal state is ``protected'' from $H_B$ by definition.  Hence for quantum systems, ``self-evidencing'' is logically equivalent to separability.  Classical memories are at risk of environmental perturbation, unless sufficient free energy is devoted to maintaining them.  Quantum ``memories'' are encoded by the dynamics $\mathcal{P}_A$ and are at risk of environmental perturbation only if separability breaks down, i.e. if $\vert AB \rangle$ becomes entangled.

With this non-spatial understanding of the ``internal'' states of a quantum system $A$ with a classical memory, both the quantum $\vert A (t) \rangle \rightarrow \vert A (t + \Delta t) \rangle$ and the classical $[ \rho_E(i) ] \rightarrow [ \rho_E(j) ]$ loops are internal to $A$.  Hence, the (classical) integrated information $\Phi (A) > 0$ as defined in Integrated Information Theory (IIT) \cite{tononi:14}, rendering $A$ ``conscious'' and hence an ``agent'' in the framework of IIT.  Hence the notion of agency in IIT is consistent with that of Definition \ref{agent-def}.

\subsection{Learning and generative models} \label{learning-gm}

Learning a generative model allowing predictions of the future behavior of $E$ is straightforward in this setting.  The prediction problem can be represented as follows:

\begin{equation}
\underbrace{[\rho_E (1)], [\rho_E (2)], \dots [\rho_E (k-1)]}_{Prior},  \underbrace{[\rho_E (k)]}_{Posterior} \rightarrow  \underbrace{[\rho_E (k+1)],}_{Prediction} \\
\end{equation}
\noindent
where $[\rho_E (1)], [\rho_E (2)], \dots [\rho_E (k-1)]$ are previous memory records, $[\rho_E (k)]$ is the current record, and $[\rho_E (k+1)]$ is the not-yet-obtained next record.  The prior data can be re-represented in three progressively more sophisticated ways:

\begin{enumerate}
\item As a probability distribution over the set of possible records, a discrete set of no more than $2^{\dim(E)}$ elements.
\item As a probability distribution over the set of pairs $([\rho_E (i)], [\rho_E (i+1)])$ and hence as a discrete Markov process.
\item As a map from probability distributions on records $1, 2, \dots i$ to probability distributions on records $1, 2, \dots i+1$ and hence as a discrete Markov kernel.
\end{enumerate}
\noindent
Representing the prior data by a discrete Markov kernel provides the greatest data compression, at the cost of more sophisticated processing by $\mathcal{P}_A$.

The process of updating a Markov kernel $\mathbb{M}^A_E (i)$ representing $A$'s prior data for $E$ can be formalized as:

\begin{equation} \label{learning}
\mathscr{L} : ( \mathbb{M}^A_E (i), ~[\rho_E (i)] ) \mapsto \mathbb{M}^A_E (i+1),
\end{equation}
\noindent
where $\mathscr{L}$ is an operator of the form $(function, data) \rightarrow ~function^{\prime}$, i.e. a learning operator.  Hence any system $A$ that stores prior information using a data structure more space-efficient than an explicit linked list can be considered to be learning.

As the records $[\rho_E (i)]$ are just bit strings and Bayesian coherence is guaranteed by the commutativity of the operators composing the QRF $\mathbf{Y}$, the sequence $[\rho_E (1)], [\rho_E (2)], \dots [\rho_E (k)]$ can be represented by a ``true'' Markov kernel $\mathbb{M}_E (k)$ satisfying the following commutativity constraint:

\begin{equation} \label{true-Markov}
\begin{gathered}
\xymatrix{
\rho_E (i) \ar[r]^{\mathbb{M}_E} & \rho_E (i+1) \\
\vert \mathscr{B} (t) \rangle \ar[u]^{\mathbf{E}}  \ar[r]^{\mathcal{P}_U} & \vert \mathscr{B} (t + \Delta t) \rangle \ar[u]_{\mathbf{E}}
}
\end{gathered}
\end{equation}
\noindent
where as before one ``tick'' of $t_A$ corresponds to $\Delta t$ externally and the notation $\vert \mathscr{B} \rangle$ indicates the state of the qubit array encoded on $\mathscr{B}$.  At step $k$ in $A$'s acquisition of state information about $E$, therefore, $A$'s prediction error $Er_E$ for $E$ is:

\begin{equation} \label{error}
Er_E (k) = d(\mathbb{M}^A_E (k), \mathbb{M}_E (k)),
\end{equation}
\noindent
where the $d$ is the metric distance on Markov kernels.  This definition is independent of $\mathscr{L}$ and hence of the sources of $A$'s prediction errors.

\subsection{Identifying and measuring systems embedded in $E$} \label{systems}

We have so far considered only observers $A$ that measure the states of their observed environments $E$ without decomposing $E$ into ``systems'' that have their own specific states.  Such undifferentiated measurements of $E$ plausibly characterize all biological systems that interact with their environments primarily biochemically, instead of mechanically \cite{robbins:16}.  A bacterium measuring an ambient salt concentration, for example, does not assign the concentration value to a specific object within the environment \cite{fgl:21}.  Animals as diverse as arthropods, cephalopods, and vertebrates, however, detect and track the states of specific external ``particles'' or objects, often other animals.  In terms of theoretical biology, this ability is either ancient, arising at least by the Cambrian explosion, or results from convergent evolution in multiple distinct lineages.

The question of how an observer $A$ distinguishes a system $S$ from the environment $\tilde{E}$, in which $S$ is embedded, is central to classical cybernetics \cite{ashby:56, moore:56} and, under the rubric of object persistence, to cognitive and developmental psychology \cite{scholl:07, fields:12}.  Here, $\tilde{E}$ indicates the remainder of $E$ when $S$ is removed, i.e. $E = S \tilde{E}$.   While the question of how an observer distinguishes an external system from its surrounding environment, prior to -- and as a precondition for -- measuring some state of interest, is often neglected by physicists, it imposes significant thermodynamic and computational requirements on observers \cite{fields:18}.  Distinguishing a system from its environment requires measurement, so it is naturally formulated in the language of QRFs.

To set up some notation, we will consider any distinct, identified system $S$ to comprise two components, $S = PR$, where $P$ is the ``pointer'' component that indicates some state $\vert P \rangle$ (or density $\rho_P$ of time-averaged samples of $\vert P \rangle$) of interest and $R$ is a ``reference'' component that by maintaining a constant state $\vert R \rangle$ (or constant density $\rho_R$ of time-averaged samples of $\vert R \rangle$) allows $S$ to be re-identified while $\vert P \rangle$ varies.  Identifying a laboratory apparatus by monitoring a time-invariant reference state provides an example; see Fig. \ref{ref-vs-pointer-fig}.

\begin{figure}[H]
\centering
\includegraphics[width=13 cm]{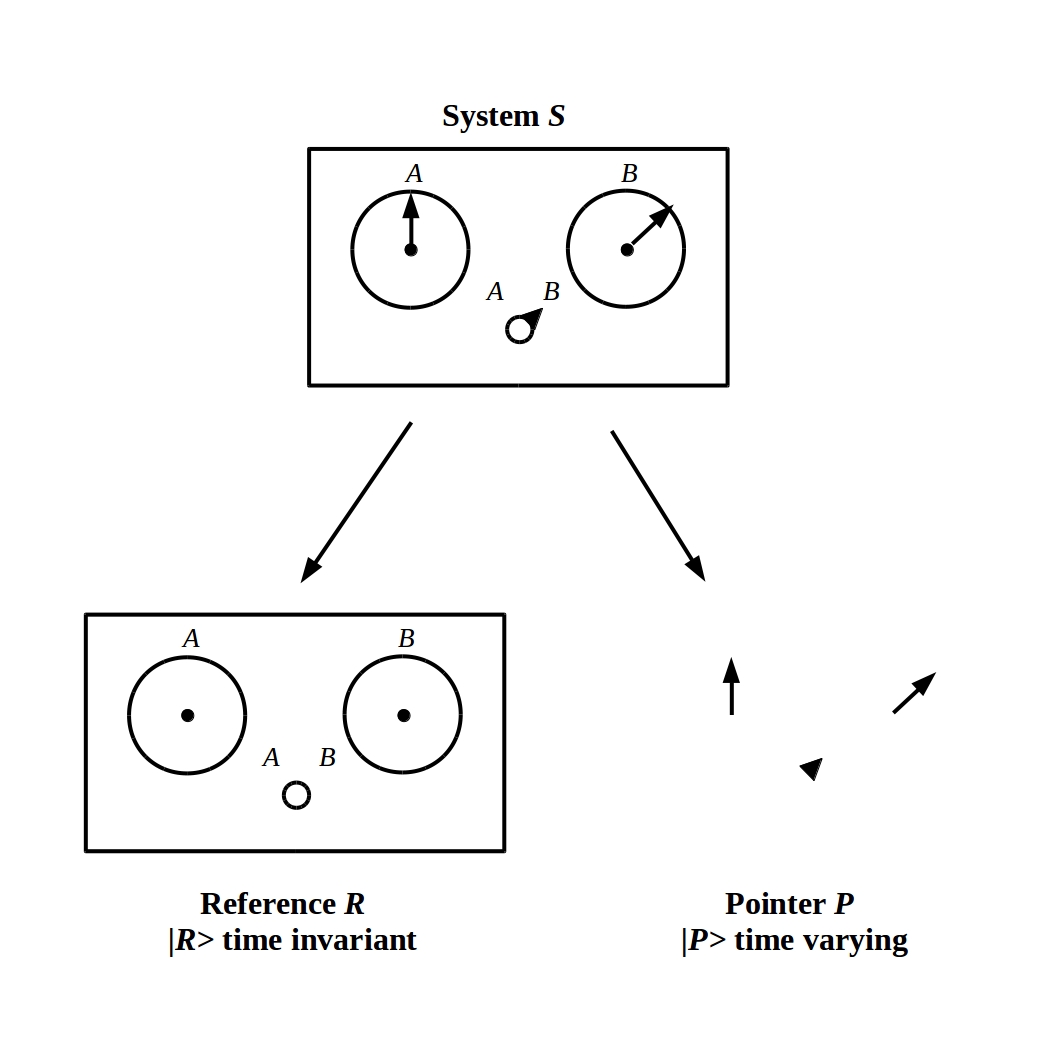}
\caption{Identifying a system $S$ requires identifying some proper component $R$ that maintains a constant state $\vert R \rangle$ (or density of time-averaged samples $\rho_R$) as the ``pointer'' state $\vert P \rangle$ (or density of time-averaged samples $\rho_P$) of interest varies.  Adapted from \cite{fg:20} Fig. 2, CC-BY license.}
\label{ref-vs-pointer-fig}
\end{figure}

Following the reasoning above, decomposing $E$ into disjoint components $E = PR\tilde{E}$ is defining subsets of operators $\{ M^E_i \} = \{ M^P_j \} \sqcup \{ M^R_k \} \sqcup \{ M^{\tilde{E}}_l \}$ where $\sqcup$ indicates disjoint union.  We can then define several QRFs:

\begin{eqnarray} \label{def-sector-QRFs}
\begin{gathered}
\mathbf{P}: P = {\rm{dom}}(\{ M^P_j \}) \rightarrow \vert P \rangle, \\
\mathbf{R}: R = {\rm{dom}}(\{ M^R_k \}) \rightarrow \vert R \rangle, \\
\mathbf{\tilde{E}}: \tilde{E} = {\rm{dom}}(\{ M^{\tilde{E}}_l \}) \rightarrow \vert \tilde{E} \rangle, \\
\end{gathered}
\end{eqnarray}
\noindent
with corresponding memory records $[\rho_P (i)]$, $[\rho_R (i)]$, and $[\rho_{\tilde{E}} (i)]$.  As QRFs that prepare as well as measure their assigned sectors of $\mathscr{B}$, these QRFs $\mathbf{P}$, $\mathbf{R}$, and $\mathbf{\tilde{E}}$ can be represented as CCCDs with the symmetric form of Diagram \eqref{cccd-2}.  The processing pathway from a given sector $X$ to its associated memory record $[\rho_X]$ can, assuming a coarse-grained memory, be represented as a CCCD with the asymmetric form of Diagram \eqref{cccd-3}; the memory-read to preparation of $X$ pathway has the opposite asymmetry.  Note that identifiability of $S$ requires $[\rho_R (i)] = [\rho_R (j)]$ for all $i, j$.  Forgetting what someone looks like, for example, can lead to re-identification failure.

Measurements of $P$, $R$, and $\tilde{E}$ are of no use unless they are recorded.  As above, Markov kernels provide the most efficient data structure.  The kernel $\mathbb{M}^A_R$ must be constant to enable system identification.  By analogy with Eq. \eqref{learning}, these kernels can be associated with a learning operators $\mathscr{L}_P$, $\mathscr{L}_R$, and $\mathscr{L}_{\tilde{E}}$; where there is no requirement that these employ the same learning algorithm.  The ``true'' Markov kernel $\mathbb{M}_E (i)$ can similarly be decomposed into components, each of which must satisfy the commutativity constraint expressed by Diagram \eqref{true-Markov}:

\begin{equation} \label{true-comps}
\begin{gathered}
\xymatrix{
\rho_P \ar[r]^{\mathbb{M}_P (i)} & \rho_P (i+1) & \rho_R (i) \ar[r]^{\mathbb{M}_R} & \rho_R (i+1) & \rho_{\tilde{E}} (i) \ar[r]^{\mathbb{M}_E} & \rho_{\tilde{E}} (i+1) \\
\vert \mathscr{B} (t) \rangle \ar[u]^{\mathbf{P}}  \ar[r]^{\mathcal{P}_U} & \vert \mathscr{B} (t+ \Delta t) \rangle \ar[u]_{\mathbf{P}} & \vert \mathscr{B} (t) \rangle \ar[u]^{\mathbf{R}}  \ar[r]^{\mathcal{P}_U} & \vert \mathscr{B} (t+ \Delta t) \rangle \ar[u]_{\mathbf{R}} &\vert \mathscr{B} (t) \rangle \ar[u]^{\mathbf{\tilde{E}}}  \ar[r]^{\mathcal{P}_U} & \vert \mathscr{B} (t+ \Delta t) \rangle \ar[u]_{\mathbf{\tilde{E}}}
}
\end{gathered}
\end{equation}
~ \\
\noindent
Hence prediction errors for $P$, $R$, and $\tilde{E}$ can be defined as in Eq. \eqref{error}.  Failing to correctly predict $\rho_R$ results in system-identification failure and renders concurrent observations of $\rho_P$ meaningless.

The ``systems'' $P$, $R$, and $\tilde{E}$ are, clearly, just subsets of outcomes obtained by measuring the qubits on $A$'s boundary/MB $\mathscr{B}$; their states are, therefore, determined by the actions of $B$'s encoding operators $M^B_i$.  As $[M^B_i, M^B_j] = 0$ for all $i, j$ by definition, $R$, $P$, and $\tilde{E}$ must all be mutually separable and hence mutually decoherent; equivalently, $\mathbf{R}$, $\mathbf{P}$, and $\mathbf{\tilde{E}}$ must all mutually commute.  Hence, $A$ can regard $R$, $P$, and $\tilde{E}$ as ``things'' with distinct identities and states as required by \cite{friston:19}.  It bears emphasis, however, that no spacetime background has been assumed in writing Eq. \eqref{ham}, in defining $\mathscr{B}$, or in defining any of the sectors $R$, $P$, or $\tilde{E}$.  The ``things'' $R$, $P$, and $\tilde{E}$ are not, therefore, observer-independent in any sense, although observers deploying similar QRFs and able to communicate (i.e. deploying QRFs that enable mutual recognition and communication) may agree as to their states \cite{mermin:18, muller:20}.  The present formalism is, therefore, ``system-free'' in the sense of \cite{grinbaum:17} and hence represents measurements using external apparatus as fully device-independent.  While device independence is implied whenever an MB is invoked \cite{clark:17}, classical formulations of measurement interactions -- indeed, of perception generally -- tend nonetheless to assume differentiated external objects {\em a priori}, as Einstein's famous insistence that the Moon is there when no one is looking exemplifies \cite{pais:79}; see \cite{marr:82, palmer:99, pizlo:01} for examples from perceptual psychophysics.  Evolutionary considerations argue against any assumption of {\em a priori} objects or even Galilean invariances of motion \cite{prakash:20, singh:21}, consistent with the background-independent approach taken here. 

\subsection{Noncommutativity and context-switching}\label{noncomm}

As emphasized by Bohr almost 100 years ago \cite{bohr:28}, a finite quantum of action $\hbar$ partitions the set of all possible quantum measurement operators into a set of ``complementary'' noncommuting pairs, the most well-known being position $\hat{x}$ and momentum $\hat{p}$.  These operators, as well as all other operators acting on ``systems'' with associated spacetime coordinates, correspond in the current background-free framework to QRFs acting on sectors of $\mathscr{B}$, i.e., on subsets of qubits as described above.  Hence, noncommuting operators correspond to noncommuting QRFs, as formalized by Theorem \ref{thm1}.

Two QRFs $\mathbf{U}$ and $\mathbf{V}$ can fail to commute only if the underlying measurement operators fail to commute.  However, as noted previously, any set of operators $M^k_i$ appearing in Eq. \eqref{ham} must all mutually commute.  Switching between noncommuting QRFs $\mathbf{U}$ and $\mathbf{V}$, therefore, entails switching between a mutually-commuting operator set $M^A_i$, of which the $M^U_j$ are a subset, and a complementary mutually-commuting operator set $O^A_i$, of which the $O^V_j$ are a subset, where for at least some $i, j, ~[M^U_i, O^V_j] \neq 0$.  This switch implements a basis rotation on $H_{AB}$, leaving its dimension $N$ and its eigenvalues, the binary representations of which are the bit strings encodable on $\mathscr{B}$, both unchanged while replacing the amplitudes $\alpha^A_i$ with amplitudes $\lambda^A_i$ so that Eq. \eqref{ham} now reads, for $A$:

\begin{equation} \label{ham2}
H_{AB} = \beta^A k_B T^A \sum_i^N \lambda^A_i O^A_i.
\end{equation}
\noindent
In practice, we will be primarily interested in partial basis rotations in which the $M^A_i$ and the $O^A_i$ substantially overlap, e.g. maintaining fixed sectors $F$, $\tilde{E}$ and $R$ while switching between complementary pointer sectors as discussed below.  Note that such basis rotations have no effect on $B$ or its operators, so are undetectable, in principle, by $B$.

An observer $A$ capable of switching between noncommuting QRFs must, to maintain an operable memory, implement a clock that is invariant under basis rotations on $H_{AB}$.  If measurements made at clock ticks $i$ and $j$ do not commute, however, the corresponding clock operations will not commute; in particular $\mathcal{G}_{ij} \circ \mathcal{G}_{ji} \neq \mathcal{G}_{ji} \circ \mathcal{G}_{ij}$ where $\circ$ is operator composition.  It is for this reason that the $\mathcal{G}_{ij}$ form a groupoid, and not a group \cite{fg:20}.  From a more practical perspective, noncommutativity forces $t_A$ to be unidirectional, and hence memory records to be encoded irreversibly with an accompanying expenditure of free energy.  The internal clock $\mathcal{G}_{ij}$ thus defines $t_A$ as an entropic time, consistent with the analysis in \cite{rovelli:21}.  Any observer $A$, therefore, observes a unidirectional flow of information from $B$ and of dissipated heat to $B$; hence any observer $A$ confirms the Second Law with respect to its internal time $t_A$.  As $A$ and $B$ are completely symmetric by Eq. \eqref{ham}, $B$ also confirms the Second Law with respect to $t_B$, showing that the Second Law is observer-relative and independent of the ``external'' time $t$, consistent with the analysis in \cite{tegmark:12}.

Switching between noncommuting QRFs while holding other QRFs constant, e.g., switching between interference (position) and which-path (momentum) measurements on an interferometer identified by a fixed reference sector $R$, is switching between mutually-noncommuting but overlapping sets of mutually-commuting operators as described in Theorem \ref{thm1} and is a canonical quantum context switch \cite{kochen:67, mermin:93}.  Pointer-state observables are, in particular, {\em observables in context} as defined in \cite{fg:21}: the state $\vert P \rangle$ of any pointer sector $P$ is measured in the context of the separable joint state $\vert R \rangle \vert \tilde{E} \rangle \vert F \rangle$, where the component $\vert F \rangle$ is unobserved by definition.   By Theorem \ref{thm1}, two sets of pointer operators $M^U_j$ and $O^V_j$ as above, that define alternative pointer sectors $U$ and $V$ by Eq. \eqref{sector}, are mutually-commuting and hence {\em co-deployable} if and only if maps $\phi$ and $\psi$ exist such that the diagram:
\begin{equation} \label{context1}
\begin{gathered}
\xymatrix@C=4pc{& &\mathbf{C} & & \\
& \mathbf{C_1} \ar[ur]^{\phi} & & \mathbf{C_2} \ar[ul]_{\psi} & \\
& U \ar[u]^{f_1} \ar[r] & \ar[ul]^{g_1} R \tilde{E} F  \ar[ur]_{g_2} \ar[r] & \ar[u]_{f_2} V &
}
\end{gathered}
\end{equation}
\noindent
commutes. Here, we have replaced the explicit operators in Diagram \eqref{double-ccd} by their corresponding sectors to simplify the notation.  We can equally well interpret $UF$ and $VF$ as contexts for the observation of sectors $R$ and $\tilde{E}$: in this case switching between $U$ and $V$—with the unobserved sector $F$ held fixed yields consistent probability distributions if and only if Diagram \eqref{context1} commutes.

If Diagram \eqref{context1} fails to commute, then pointer-state observables are said to be {\em non-co-deployable}. Non-commutativity of the CCD in \eqref{context1} had been specified in \cite[Th. 7.1]{fg:21} in terms of non-existence of a consistently definable joint probability distribution of conditionals, such as for Diagram \eqref{ccd-prob-1}. This non-co-deployability of observables thus amounts to occurrence of intrinsic (quantum) contextuality in relationship to \eqref{context1}.\footnote{As recalled in e.g. \cite[\S3]{sulis:21}, `non-commutativity' is at the very heart of contextuality, as first formulated by von Neumann in terms of non-commutativity of self-adjoint operators representing measurement, with the impossibility of simultaneously measuring the eigenvalues corresponding to non-commuting operators. In summarizing the principal results of ensuing hidden variables theory, Mermin \cite{mermin:90} demonstrated the impossibility of finding a joint probability distribution for all possible observables.}

We will see in \S\ref{sources} below that context switching increases variational free energy (VFE) by generating Bayesian``prediction errors"; hence context-switching makes minimizing VFE and hence complying with the FEP more difficult.  Deploying noncommuting QRFs against a fixed background can, however, lead to radically better generative models, as the history of technological applications of quantum theory attests.  Hence, context-switching poses a fundamental challenge to any classical formulation of the FEP, and a fundamental explanadum for a quantum formulation.

\section{FEP for generic quantum systems} \label{FEP}

\subsection{Defining VFE for quantum systems} \label{VFE}

The FEP is a variational or least-action principle: it states that a system enclosed by an MB, and therefore having internal states that are conditionally independent of its environment, will evolve in a way that tends to minimize a VFE that is an upper bound on surprisal.  Formally, the VFE $F(\pi)$, where $\pi$ is a ``particular'' state $\pi = (b, \mu)$ comprising MB $(b)$ and internal $(\mu)$ components, can be written \cite[Eq. 8.4]{friston:19},

\begin{equation} \label{VFE-def}
F(\pi) \triangleq \underbrace{\mathfrak{I}(\pi)}_{Surprisal} + \underbrace{D_{KL}[Q_{\mu} (\eta ) \parallel P(\eta | b)]}_{Divergence} \geq \mathfrak{I}(\pi).
\end{equation}
\noindent
This functional as an upper bound on surprisal $\mathfrak{I}(\pi) = -\log P(\pi)$ because the Kullback-Leibler divergence ($D_{KL}$) term is always non-negative. This KL divergence is between the density over external states $\eta$, given the MB state $b$, and a variational density $Q_{\mu} (\eta )$ over external states parameterised by the internal state $\mu$.  If we view the internal state $\mu$ as encoding a posterior over the external state $\eta$, minimizing VFE is, effectively, minimizing a prediction error, under a generative model supplied by the NESS density. In this treatment, the NESS density becomes a probabilistic specification of the relationship between external or environmental states and particular (i.e. ``self'') states.

In the notation developed in \S \ref{2} and \ref{3} above, we can write the surprisal for a quantum system $A$ in its most general form as:

\begin{equation} \label{surprisal-gen}
\mathfrak{I}^A (t) = -P(\vert \mathscr{B} (t) \rangle ~\vert~ |A (t) \rangle )
\end{equation}
\noindent
and the corresponding evidence bound as:

\begin{equation} \label{eb-gen}
D_{KL}[Q_{\vert \mathscr{B} (t) \rangle} (\vert B (t) \rangle ) \parallel P(\vert B (t) \rangle ~\vert~ |\mathscr{B} (t) \rangle ].
\end{equation}
\noindent
In the current setting, however, these expressions have little direct utility, as our effective starting point, Eq. \eqref{ham}, constrains neither $\vert A (t) \rangle$ nor $\vert B (t) \rangle$.  Indeed, from a strict formal perspective, neither Eq. \eqref{surprisal-gen} nor Eq. \eqref{eb-gen} is well-defined in the current setting.  The full boundary/MB state $\vert \mathscr{B} (t) \rangle$ is, moreover, not an observable for $A$ (or $B$), as the thermodynamic-resource sector $F$ remains unobservable by definition.  We have, however, already derived in \S\ref{learning-gm} a representation of $A$'s prediction error, Eq. \eqref{error}, which we reproduce here for reference:

\begin{equation} \label{error2}
Er_E (k) = d(\mathbb{M}^A_E (k), \mathbb{M}_E (k)).
\end{equation}
\noindent
In this expression, the timestep $k$ counts $A$'s internal clock time $t_A$ and the kernels $\mathbb{M}^A_E$ and $\mathbb{M}_E$ are derived from observables and therefore constrained by the theory.  The kernel $\mathbb{M}_E (k)$ represents the observable behavior of $\mathscr{B}$, as localized to the sector $E$, up to $t_A = k$.  The kernel $\mathbb{M}^A_E (k)$ is $A$'s generative model of the action of the unknown dynamics $\mathcal{P}^B (t)$ on $\mathscr{B}$, also as localized to $E$.  Hence $Er_E (k)$ represents $A$'s total {\em reducible} uncertainty about $B$ at $t_A = k$.  It is, therefore, an upper bound on surprisal analogous, in the current setting, to $F(\pi )$.

The operators $M^E_i$ referenced by Eq. \eqref{error}, i.e. \eqref{error2} must, clearly, all be co-deployable.  In practice, however, $E$ is as discussed above subject to context switches of the form $UR\tilde{E} \rightarrow VR\tilde{E}$ whenever $A$ switches between noncommuting pointer QRFs $\mathbf{U}$ and $\mathbf{V}$ and hence non-co-deployable operators $M^U_i$ and $M^V_i$.  Hence $Er_E$ is only well-defined in the absence of context switches; in the presence of context switches the generalized uncertainty relation:

\begin{equation}
\Delta u \Delta v \geq \hbar /2
\end{equation}
\noindent
for pointer outcomes $u$ and $v$ can generate divergent uncertainties.  Hence in practice, any system $A$ is faced with separately minimizing:

\begin{equation} \label{error3}
Er_X (k) = d(\mathbb{M}^A_X (k), \mathbb{M}_X (k)).
\end{equation}
\noindent
for each sector $X$ defined by a QRF $\mathbf{X}$.  We can, therefore, formulate the FEP for generic quantum systems, taking context-switching into account, as:

\begin{quote}
{\bf FEP}: {\em A generic quantum system $A$ will act so as to minimize $Er_X$ for each deployable QRF $\mathbf{X}$.}
\end{quote}
\noindent
A trivial agent can be viewed as executing a trivial QRF, i.e. as only exercising choice of basis for writing to and reading from $\mathscr{B}$ as a whole, and so satisfies the FEP trivially.

\subsection{Sources of VFE for quantum systems} \label{sources}

As noted in the Introduction, there is no source of objective randomness in the current formalism.  Indeed, an observer $A$ can be regarded as {\em certain} of the states $\vert \tilde{E} \rangle$, $|R \rangle$, $\vert P \rangle$, and $\vert Y \rangle$ of the observable sectors of $\mathscr{B}$ at every (external) time $t$.  Uncertainty and prediction error -- and hence, VFE -- is generated in the current formalism by $A$'s in-principle ignorance of both the state $\vert B \rangle$ and the dynamics $\mathcal{P}_B$ of its interaction partner $B$. As the bits $A$ reads from $\mathscr{B}$ are written by $\mathcal{P}_B$, $A$'s ability to predict the future states of its observable sectors, and hence to minimize $Er_X$ for each sector $X$ via Eq. \eqref{error3}, depends on its ability to predict the behavior of $\mathcal{P}_B$ locally on each observable sector.  As the thermodynamic sector $F$ is not observed, direct predictions on $F$ are not possible; the local behavior of $\mathcal{P}_B$ on $F$ can at best be predicted from its local behavior elsewhere.  An animal, for example, must employ its available senses -- hence its observable sectors -- to predict the nutritional value of food.

The option space governing $A$'s ability to locally predict $\mathcal{P}_B$ is summarized in Fig. \ref{four-options-fig}.  What is important for $A$ is not the dynamic complexity or even the dimension of $\mathcal{P}_B$, both of which are unobservable in principle, but rather the dynamic complexity of the action of $\mathcal{P}_B$ on $\mathscr{B}$ (the dimension of this action is, clearly, just the dimension of $\mathscr{B}$).  Here, the weak-interaction limit that allows separability between $A$ and $B$ is significant: $H_{AB}$ (and hence $\mathscr{B}$) must have significantly lower dimension that $H_B$ (and hence $\mathcal{P}_B$) if the weak interaction limit to is hold.  The simplest case is shown in Fig. \ref{four-options-fig}, Panel a), in which the system $B$ is a trivial agent deploying no QRFs other that the choice of basis for interactions with $\mathscr{B}$.  The action of $\mathcal{P}_B$ is, in this case, limited to choice of basis, e.g., to rotating the $z$ axis $z_B$ in Fig. \ref{qubit-screen-fig}.  As discussed in \S\ref{rf1}, basis rotation by $B$ generates quantum noise in the communication channel defined by $H_{AB}$ that is indistinguishable by $A$ from classical noise.  Hence, the trivial agent $B$ in Fig. \ref{four-options-fig}, Panel a) ``looks like'' a noise source to $A$.  Emission of Hawking radiation from a black hole (BH) provides perhaps the most pure example of such a noise source; while the dimension ``inside'' the BH can be arbitrarily large (see e.g. \cite[Fig. 19]{almheiri:21}), the internal dynamics are uncoupled from the classical information encoded on the horizon and hence have no classical computational power.  As will be discussed in \S\ref{asymptotic} below, a ``small'' trivial agent will be driven by the FEP toward entanglement with the larger system $A$.

\begin{figure}[H]
\centering
\includegraphics[width=15 cm]{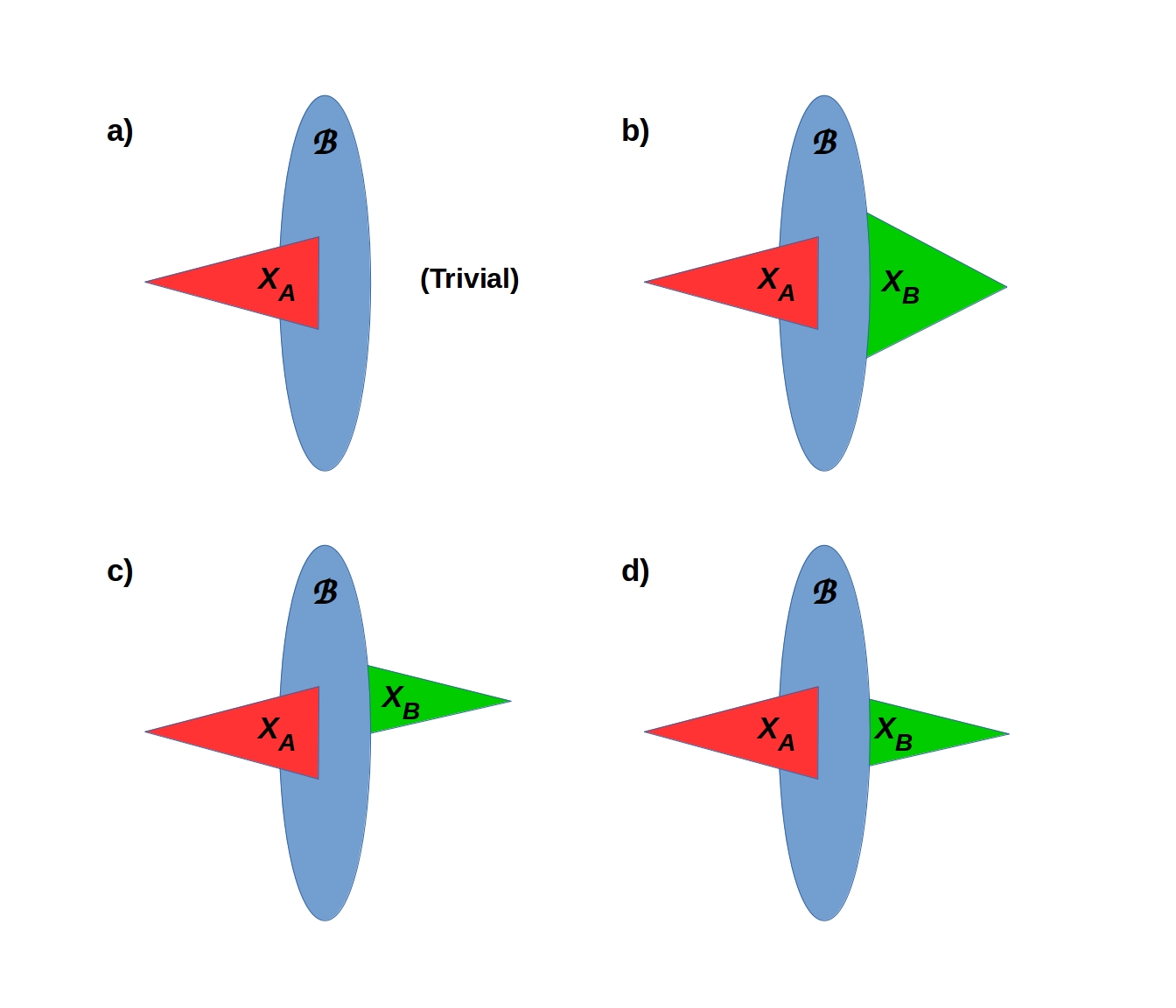}
\caption{Four options for $A$'s ability to predict the local behavior of $\mathcal{P}_B$ on an observable sector $X_A$.  a) A trivial agent deploying no QRFs beyond choice of basis for interacting with $\mathscr{B}$ appears as a noise source to $A$.  b) $B$ encodes a sector $X_B$ that contains $X_A$; the bits on $X_B$ but outside $X_A$ encode ``nonlocal hidden variables'' for $A$.  c) The sectors $X_A$ and $X_B$ overlap; the areas of non-overlap become noise sources.  d) If $X_A = X_B$, VFE is generated by insufficient learning.}
\label{four-options-fig}
\end{figure}

The more interesting parts of $A$'s option space for prediction are shown in Fig. \ref{four-options-fig}, Panel b), c), and d), in which $B$ is nontrivial.  If $B$ is nontrivial, it deploys at least one QRF $\mathbf{X_B}$ acting on a sector $X_B$.  As discussed in \S\ref{breaking}, $B$'s sectors must be mutually decoherent, so the action of $\mathcal{P}_B$ on $X_B$ is independent of its action elsewhere; it is this independence that makes prediction possible.  If $X_B$ does not overlap any {\em observable} sector for $A$, however, $B$ will appear trivial, i.e. as a noise source, to $A$.  Hence, the interesting cases are the ones in which $A$'s and $B$'s observable sectors overlap; this is the case, intuitively, in which $A$ and $B$ can ``see each other'' and hence interact in the ordinary, nontechnical sense of that term.

In Fig. \ref{four-options-fig}, Panel b), $B$'s sector $X_B$ fully contains $X_A$.  As noted above in connection with  Eq. \eqref{def-sector-QRFs}, all measure -- prepare QRFs $\mathbf{X}$ are symmetric, i.e ${\rm{codom}}(\mathbf{X}) = {\rm{dom}}(\mathbf{X})$.  Preparation by $B$ of the bits in $X_A$ will, therefore, in general depend on bits outside of $X_A$ but within $X_B$, i.e. on bits with the remainder $X_B \setminus X_A$.  The values of these bits are ``nonlocal hidden variables'' \cite{mermin:93} from $A$'s perspective; they affect what is observed on sector $X_A$ without being local to, i.e., contained within, $X_A$.  Indeed, such bits may be within $F$ and hence unobservable in principle by $A$.  Changes in the values of these nonlocal hidden variables are, effectively, context changes as defined in \cite{dzha:17b, dzh:18}; probability distributions $P(X_A | \zeta)$ and $P(X_A | \xi)$ for distinct hidden-variable states $\zeta$ and $\xi$ may be different.  The ``context-blind'' distribution $P(X_A)$ can, in this case, fail to be well-defined over time.  Such failures manifest as violations of Leggett-Garg inequalities \cite{emary:13}, i.e. as ``quantum hysteresis'' effects due to nonlocal (i.e. outside of $X_A$) and possibly unobservable causes.  In the language of artificial intelligence or robotics, they appear as failures to solve the Frame Problem \cite{fg:21}, the problem of predicting what will not change as the result of an action \cite{mccarthy:69}.  If $X_A$ is a reference sector for $A$, Frame Problem solution failures on $X_A$ can result in failures of object re-identification \cite{fields:13}. 

A situation in which $X_A$ fully contains $X_B$ presents similar issues to that in Fig. \ref{four-options-fig}, Panel b), except here the ``hidden variables'' are in $X_A \setminus X_B$ and hence are accessible to $A$.  The bits in $X_A \setminus X_B$ nonetheless contribute VFE -- effectively, noise -- to $X_A$ that is unconstrained by $X_B$.  If $X_A$ and $X_B$ overlap with remainders, as in Fig. \ref{four-options-fig}, Panel c), a similar noise contribution to $X_A$ (or on $B$'s side, to $X_B$) results.  The final possibility is, clearly, that in which $X_A = X_B$ as shown in Fig. \ref{four-options-fig}, Panel d).  Here, the source of VFE is not noise, but rather differences in the computations implemented by the QRFs $\mathbf{X_A}$ and $\mathbf{X_B}$.  Such differences correspond, in the notation of \S\ref{chan-th}, to differences in the structures of the CCCDs implementing $\mathbf{X_A}$ and $\mathbf{X_B}$, e.g. differences in the ``connection weights'' if these are thought of as ANNs or VAEs.  They correspond, in other words, to learning failures, e.g. due to insufficient training-set representativeness, as Eq. \eqref{learning} renders obvious.  We can, therefore, represent the overall situation for any observer $A$ as:
\begin{equation} \label{VFE-sources}
\mathrm{VFE ~=~ Noise ~+~ Insufficient ~Learning}
\end{equation}
\noindent
consistently with Eq. \eqref{VFE-def} above.  Here ``noise'' includes VFE generated by unobserved context changes (Leggett-Garg violations or Frame Problem solution failures) as well as ``classical'' noise. 

\subsection{Asymptotic behavior of the FEP} \label{asymptotic}

Having seen how VFE is generated, we can now ask how it is minimized: how, in other words, an agent acts in accordance with the FEP.  As discussed in the Introduction, the FEP in its classical form is effectively the statement that any system with sufficient stability to be a ``thing'' -- i.e. a system with a [quasi-] NESS density -- will act so as to preserve its ``thingness'' by maintaining the integrity of its MB and hence the integrity of its ``self'' as distinct from its surroundings.  Comparing Eq. \eqref{VFE-def} and \eqref{VFE-sources}, it becomes clear what this amounts to: a system ``self evidences'' by behaving in a way that minimizes noise while improving learning.  This, as is well known, induces a trade-off: learning requires seeking uncertainty in order to minimize it.  As emphasized in \cite{oudeyer:16}, successful learning requires a focus on learnable tasks and avoidance of unlearnable tasks.  Hence, minimizing VFE is removing all removable noise.  It does not, in practice, lead to perfect predictability of future MB states, but to best feasible predictability of future MB states. In classical FEP formulations, this becomes clearly evident in the form of a functional called expected free energy (EFE); namely, VFE expected under a posterior predictive density that is conditioned on action. In this setting, the most probable actions are those that minimise EFE and thereby resolve uncertainty or maximise information gain (a.k.a., intrinsic motivation or epistemic affordance). In short, novelty-seeking is an emergent property of any``thing", under the FEP.

We can, however, ask in the current framework how the FEP behaves asymptotically, i.e., what are the consequences for $A$ as, in the notation of Eq. \eqref{error3}, $Er_X (k) \rightarrow 0$ as $k \rightarrow \infty$ for all observable sectors $X$.  This clearly involves implementing a learning operator $\mathscr{L}_X$ for each sector $X$ that is capable of asymptotically-perfect learning: let us assume this is the case.  What remains given Eq. \eqref{VFE-sources} is noise, including noise due to observed context switches.  Only one mechanism for removing noise is available: that shown in Fig. \ref{four-options-fig}, Panel d).  Hence we can conclude:
\begin{quote}
{\em The FEP asymptotically drives alignment of QRFs across $\mathscr{B}$.}
\end{quote}
\noindent
It drives any observer $A$, in particular, to match any context switches by its interaction partner $B$ in order to maintain QRF alignment.  Let us now consider, therefore, a situation in which {\em all} QRFs deployed by $A$ and $B$ are aligned as in Fig. \ref{four-options-fig} d), and in which {\em all} of $A$'s QRFs have learned the local behavior of $\mathcal{P}_B$ on their sectors perfectly.  The local behavior of $\mathcal{P}_B$ on some shared sector $X$ is determined by $B$'s QRF $\mathbf{X_B}$.  This QRF $\mathbf{X_B}$ is, however, a quantum computation; as such, it encodes nonfungible -- not finitely classically encodable -- information as shown in \cite{bartlett:07}.  The future behavior of $\mathbf{X_B}$ can, therefore, only be perfectly predicted by $\mathbf{X_B}$ itself, that is:
\begin{equation} \label{equal-QRFs}
Er_X \rightarrow 0 ~\Rightarrow~ \mathbf{X_A} = \mathbf{X_B}
\end{equation}
\noindent
If $A$ and $B$ are separable, the consequent in Eq. \eqref{equal-QRFs} violates the no-cloning theorem \cite{wooters:82}: it demands that the internal quantum state $\vert B\vert_X (t) \rangle$, the time (external $t$) evolution of which implements $\mathbf{X_B}$, be replicated exactly in $A$.  Hence, if Eq. \eqref{equal-QRFs} holds, $A$ and $B$ cannot be separable.  Therefore we have:
\begin{quote}
{\em If $AB$ is isolated, the FEP asymptotically drives the joint state $\vert AB \rangle$ to entanglement.}
\end{quote}
\noindent
The claim that isolated systems are driven to entanglement is familiar from our initial discussion of bipartite interactions in \S\ref{holographic-1}: all isolated systems are driven to entanglement by unitary evolution.  Separability is a special case, an approximation that holds only under conditions of weak interactions and short observation times.  Hence, we can, finally, conclude:
\begin{quote}
{\em The FEP is, asymptotically, the Principle of Unitarity.}
\end{quote}
\noindent
The FEP is, in other words, asymptotically equivalent to the first axiom of quantum theory.  Any quantum system, therefore, {\em must} behave in accord with the FEP; doing so is simply approaching entanglement with its environment as required by the Principle of Unitarity.  The FEP is therefore, consistent with the results obtained in \cite{friston:19}, a fundamental, generic physical principle.  Conversely, the Principle of Unitarity -- the principle that observable information is conserved -- can now be seen as a fundamental principle of cognitive science.

As an example of the FEP in action, let us return to the situation in which a ``small'' trivial system $B$ interacts with a larger, nontrivial system $A$ considered above.  The only freedom $B$ has in this case is freedom of basis choice for reading and writing from $\mathscr{B}$.  From $B$'s perspective, any basis choice that is misaligned with $A$'s basis choice generates noise.  The FEP will, therefore, drive $B$ to align its choice of basis -- $z$ axis in Fig. \ref{qubit-screen-fig} -- with that of $A$.  This basis alignment, however, leaves $B$ entangled with $A$.  This is not surprising.  When a photon, for example, interacts with an atom, it is completely absorbed and loses its identity as a ``thing''; its state becomes irreversibly entangled with that of the atom with which it interacts.  

We can develop this result more formally as follows, noticing that a preparation operation by $\mathbf{B_X}$ followed by a measurement operation by $\mathbf{A_X}$ can be described by a CCCD in the form of Diagram \eqref{cccd-1}, with the classifiers $\mathcal{A}_i$ identified with the bits comprising sector $X$ that are prepared and then measured.  The question of asymptotic behavior is then the question of whether this CCCD is symmetric, i.e. whether the cores $\mathbf{C^{\prime}} = \mathbf{D^{\prime}}$ in Diagram \eqref{cccd-1}.  We approach this by considering the Markov kernels $\mathbb{M}^A_X$ and $\mathbb{M}_X$, the difference between which defines the error $Er_X$ (via Eq. \eqref{error3}) that the FEP asymptotically sends to zero. 

As shown in \cite[p.17]{pratt:97}, following \cite[\S 4]{pratt:96}, every category $\mathfrak{C}$ whose arrows form a set $K$ embeds fully as a subcategory into $\mathbf{Chu}$.   There is, therefore, an induced mapping $\mathscr{F}: \mathfrak{C} \longrightarrow \mathbf{Chu}$ realizing $\mathscr{F}(\mathfrak{C})$ as an embedded subcategory of $\mathbf{Chu}$ that consists of some objects of the latter, together with all of the arrows between them (for background on such embeddings, see e.g. \cite{adamek:04, awodey:10}).

Now we recall the direct association between Markov kernels and conditional probabilities, and apply the above embedding to the category $\mathfrak{P}$ of conditional probabilities, following mainly \cite{culbertson:13, culbertson:14} (cf. \cite{giry:82, lawvere:62}). Objects of $\mathfrak{P}$ consist of countably generated measurable spaces $(\mathcal{X}, \Sigma_{\mathcal{X}})$, and arrows of $\mathfrak{P}$ between two such objects, having the form:
\begin{equation}\label{markov-ker-1}
\mathbb{M}: (\mathcal{X}, \Sigma_{\mathcal{X}}) \longrightarrow (\mathcal{Y}, \Sigma_{\mathcal{Y}}).
\end{equation}
\noindent
These arrows are Markov kernels assigning to each $x \in \mathcal{X}$, and each measurable set $Q \in \Sigma_{\mathcal{Y}}$, the probability of $Q$ given $x$, denoted here by $\mathbb{M}(Q \vert x)$, whenever defined.
We can also write this as $p_{\mathbb{M}}(Q, \vert x)$, the (regular) conditional probability determined by the arrow $\mathbb{M}$, `regular' in so far that $\mathbb{M}$ is conditioned on points rather than on measurable sets in $\Sigma_{\mathcal{X}}$.  As arrows given by Eq. \eqref{markov-ker-1} specify a Markov process, they comprise a semigroup, and therefore a set.  Hence, we obtain a full embedding $\mathfrak{P} \longrightarrow \mathbf{Chu}$, where the (arrow) set $K$ is identified with a set of conditional probabilities $K = \{p_{\mathbb{M}}(S, \vert x)\} \subseteq [0,1]$.\footnote{In a similar way, Abramsky in \cite{abramsky:12} shows that the Dirac-von Neumann formulation of quantum mechanics can be conveniently represented in the category $\mathbf{Chu}$.}  Using the fact that $\mathbf{Chu}$ and the Channel Theory category $\mathbf{Chan}$ are isomorphic categories with respect to their respective objects and arrows, we can summarize as follows:
\begin{proposition}\label{markov-embed-1}
The category $\mathfrak{P}$ embeds fully into $\mathbf{Chan}$ with classification $\Vdash_{\mathcal{A}}$ realized by the conditional probability $p_{\mathbb{M}}(\cdot \vert  \cdot)$ (the Chu space valuation/satisfaction relation) whenever this is defined.
\end{proposition}
Without loss of generality, we can assume the CCD in Diagram \eqref{ccd-prob-1} to be a diagram in this embedded subcategory (for a survey of probability spaces and Bayesian belief networks in terms of the categories $\mathbf{Chu}$ and $\mathbf{Chan}$, see \cite[\S2]{fg:21}).

We can now proceed to consider the asymptotic behavior of the FEP, in particular the conditions under which $\vert \mathbb{M}_X - \mathbb{M}^A_X \vert \rightarrow 0$ for an arbitrary system $A$ and sector $X$ (with QRF $\mathbf{X}$), in the setting where $A$ interacts with some $B$ and the joint system $AB$ is isolated. To make sense of the formal difference $\mathbb{M}_A - \mathbb{M}_B$, that is, to make sense of the difference between two arrows in $\mathfrak{P}$, we can use the metric distance on Markov kernels as in Eq. \eqref{error}. This distance is then manifestly the difference between the conditional probabilities $p_{\mathbb{M}_A}(\cdot \vert \cdot) - p_{\mathbb{M}_B}(\cdot \vert \cdot)$, in the way the corresponding Markov kernels determine these as described above. As discussed previously, this is a difference in the (conditional) probabilities of observable behavior in sectors. Thus, $\vert \mathbb{M}_A - \mathbb{M}_B \vert \rightarrow 0$ is interpreted as the metric distance $d(\mathbb{M}_A,\mathbb{M}_B) \rightarrow 0$ in the asymptotic limit.

With the labeling by $A,B$, let us return to \eqref{ccd-prob-1} which we adopt to provide two separate CCDs as specified by \eqref{ccd-prob-1} denoted by $\rm{CCD}_A$ and $\rm{CCD}_B$. The asymptotic limit $\rightarrow 0$ of the metric distance of the kernels determining the difference of the conditionals (hence $\rightarrow 0$ in the asymptotic limit) provides the sense in which the diagrams $\rm{CCD}_A$ and $\rm{CCD}_B$, along with their respective colimits $\rm{colim}(\mathbf{C}'_A)$ and $\rm{colim}(\mathbf{C}'_B)$, can be identified in this limit (or in the notation of Diagram \eqref{cccd-1}, the cores $\mathbf{C^{\prime}}$ and $\mathbf{D^{\prime}}$ as required above become identical).

\section{Discussion} \label{disc}

\subsection{High-level overview} \label{overview}

While much of the foregoing has been technical, it is conceptually straightforward.  To briefly review, the classical FEP as developed in \cite{friston:19} considers the joint environment-agent system as a random dynamical system that possesses an attracting set.  By placing particular constraints on the coupling among systemic states (e.g., with sparse coupling in flow operators or stochastic differential equations), one can partition the joint state space into external, internal and blanket (MB) states. In turn, the blanket states are partitioned into sensory states that mediate the influence of external states on internal states and active states that mediate the influence of internal states on external states. Crucially, this particular partition imposes conditional independence between internal and external states, given the blanket states. The final move in the classical FEP is to induce a variational density $Q_{\mu} (\eta )$ over external states that is parameterised by internal states. It is then fairly straightforward to show that the expected flow of internal (and active) states can be expressed as a gradient flow on a variational free energy. This free energy is effectively the divergence between the variational density encoded by internal states and the density over external states conditioned on the blanket states. This licences an interpretation of internal and active states in terms of active inference, or a Bayesian mechanics, in which their expected flow can be read as perception and action, respectively. In other words, active inference is a process of Bayesian belief updating that incorporates active exploration of the environment.

Reformulating the FEP within quantum information theory allows us to drop the assumption of an observer-independent spacetime background characterized by (continuous) frames of reference that can be specified to infinite precision, and also to drop the assumption of randomness. In the quantum formulation, the blanket states are implemented by a holographic screen separating the interacting systems $A$ and $B$.  The screen is the (topological) locus of the interaction $H_{AB}$; ``sensory'' and ``active'' states of the classical MB become incoming and outgoing encodings of bits on the screen.  The interaction is symmetrical across the screen: the reciprocal exchange between ``internal'' $A$ and ``external'' $B$ systems takes the form of answers and questions (formulated as Hermitian operators $M^k_i$), where questions correspond to classical action (mediated by active states) and answers correspond to classical perception (mediated by sensory states).  There is, crucially, no assumption that $A$ and $B$ share preparation and measurement bases; basis mismatches between $A$ and $B$ generate quantum noise that is ``perceived'' as classical noise.

An ``agent'' in the quantum formalism is a system that deploys quantum reference frames -- effectively, {\em concepts} that identify persistent objects and their time-varying states -- when interacting with its environment.  Deploying distinct QRFs breaks the thermodynamic symmetry of the screen for each agent, redirecting energy flows within each agent to fund processing and recording memory records of some bits while others -- those in the ``thermodynamic'' sector $F$ -- remain necessarily unobserved.  Thus the quantum formalism explicitly enforces a distinction that the classical formalism leaves implicit: that between observed and unobserved parts of the environment.  Mismatches between QRFs deployed by two interacting agents generate noise and prediction errors, including incorrect Frame Problem solutions and failures to correctly re-identify objects.

In the classical setting, agents ``self-evidence'' by maintaining their nonequilibrium steady-states or, equivalently, the integrity of their MBs.  From a mathematical point of view, this is maintaining their identities (in the sense of being independently well-defined) as systems.  In the quantum setting, being independently well-defined is being {\em separable} from -- not entangled with -- the environment; hence ``self-evidencing'' is maintaining separability.  The FEP identifies self-evidencing with the minimization of Bayesian prediction error: to``be" is to be capable of successful predictions; sometimes described as ``predicting yourself into existence".  Minimizing prediction error is, in the quantum setting, minimizing the difference between two Markov kernels.  As in the classical setting, noise and errors due to insufficient learning must both be minimized, a process that requires trade-offs between the two.

The classical and quantum formulations of the FEP differ in their asymptotic behavior, i.e., as total prediction error is driven toward zero.  A classical ``perfect predictor'' achieves maximal classical correlation and hence maximal behavioral synchrony with its environment; it becomes a ``perfect regulator'' in the sense of the good regulator theorem \cite{conant:70}.  In a quantum setting, perfect prediction entails shared QRFs between interaction partners and hence entanglement; a quantum ``perfect regulator'' becomes indistinguishable from the environment it is predicting.  While this difference in asymptotic behaviors is a formal, mathematical outcome, it can be traced to a difference in fundamental assumptions.  The classical formalism assumes a spacetime background, and hence can rely on separation in space to distinguish between systems.  The quantum formalism is background free: space is simply an observable, represented by a QRF that a system may or may not deploy.  It can, therefore, play no ``ontic'' role in maintaining distinctions between systems.  This reflects the general role of (``physical'' 3d) space in quantum field theories: space is there to enforce separability (see \cite{addazi:21} for a general discussion of this point from a gauge-theoretic perspective).

Both classical and quantum formulations of the FEP engender a fundamental form of {\em epistemic} solipsism, in the sense that the coupling between internal and external systems precludes the states of one from ever ``knowing'' the states of the other.  While this seems counter-intuitive, it is a straightforward consequence of the use of vector spaces to represent physical states.  Vector-space product operators -- the classical Cartesian product or the Hilbert-space tensor product -- are by definition associative (equivalently, dynamic operators such as Hamiltonians are additive); product decompositions are, therefore, undetectable across decompositional boundaries in any vector space \cite{fields:16}.  This sense of epistemic solipsism does {\em not}, as Fuchs emphasizes \cite{fuchs:10}, in any way suggest ontic or metaphysical solipsism.  Both classical and quantum formulations of the FEP -- indeed, any theory of measurement or observation -- requires two interacting systems, observer and observed.  The idea of a metaphysically solipsist theory of observation is self-contradictory. 

\subsection{Summary of results} \label{results}

We have obtained three results in this paper:
\begin{enumerate}
\item Given the standard free-choice assumption, the intuitive idea of an ``agent'' or IGUS can be fully formulated within background-independent, scale-free quantum information theory.
\item The FEP can be given a quantum-theoretic formulation that renders it applicable to generic quantum systems.
\item When formulated as a generic principle of quantum information theory, the FEP is asymptotically equivalent to the Principle of Unitarity.
\end{enumerate}
\noindent
Result 1) places the long-standing practice of treating generic quantum systems as ``observers'' on a firm theoretical foundation.  It allows a formal specification, in the language of QRFs, of exactly what systems a given observer is capable of recognizing and what states of those systems it is capable of measuring.  Such specifications require no ``ontic'' assumptions of observer-independent systems; hence they are compliant with a device-independent theoretical strategy.  Result 1) also provides a formal definition of an ``agent'' in the context of the free-choice assumption, and provides a generic language -- the category-theoretic language of Channel Theory -- for specifying the semantics assigned by an agent to an observational outcome.  Finally, Result 1) shows how semantics arise from thermodynamic symmetry breaking on a holographic screen, and provides a formal mechanism for quantifying energy flows that enable classical computation and classical memory encoding.

Result 2) extends the range of applicability of the FEP to generic quantum systems independently of spacetime background or scale-dependent assumptions.  Quantum fields, black holes, and other spatially-distributed or topologically-defined systems can, therefore, be regarded as Bayesian observers and, if free choice is assumed, as Bayesian agents.  Result 2) therefore makes clear the sense in which generic quantum systems can be regarded as ``users'' of quantum information theory, as proposed under the rubric of QBism \cite{fuchs:13, mermin:18, fuchs:10}.  Indeed Result 2) renders QBism a consequence of quantum information theory, not an interpretation.

Result 3) shows that the FEP is compliant with the Principle of Unitarity and, conversely, that unitary evolution is compliant with the FEP.  It allows us, in particular, to understand quantum context-switching as both a source of prediction errors and a strategy for reducing prediction errors.  Result 3) also allows us to view separability as a resource for classical communication and computation that is analogous to entanglement as a resource for quantum communication and computation.  Hence, it allows us to consider trade-offs between these resources by systems that maintain approximate separability while also employing shared entanglement.  Such systems have been studied in the abstract, and can potentially exceed the computational power of Turing machines (e.g. can solve the Halting problem as shown in \cite{ji:21}).  Result 3) suggests that (some) biological systems may have this capability, as discussed further in \S\ref{pred} below. 

\subsection{Applications to biological cognition} \label{biology}

In addition to the results listed above, the current framework has a variety of more specifically biological consequences, some of which have been discussed already in \cite{fgl:21}.  It predicts, for example, that moving in ordinary 3d space does not require a QRF for Euclidean space and hence, does not require an experience of space.  From a classical perspective, this is certainly true; as evidenced e.g. by place and grid cells in mammalian brains \cite{moser:15, stachenfeld:17, whittington:19} that appear to encode a coarse-grained representation of location, head-direction etc., in various frames of reference. This coarse graining endorses another prediction; namely, that actionable classical encodings are coarse-grained. Any system that encodes information irreversibly is, therefore, faced with a choice that its computational architecture must resolve: the trade-off between preserving information in memory and losing information due to coarse graining. On a classical FEP account, this coarse graining is mandated by the minimization of VFE, which can be expressed as complexity minus accuracy. Classically, complexity corresponds to the KL divergence between posterior and prior beliefs as in Eq. \eqref{VFE-def}. This relative entropy clearly depends upon the degree of discretisation or coarse graining afforded by the dimensionality of internal states. Effectively, this KL divergence scores the computational and thermodynamic cost of belief updating that is mitigated by coarse graining or, in quantum terms, devoting memory resources to the results computed by only some QRFs, possibly in a context-sensitive fashion \cite{lloyd:13}.

The ubiquity of context-dependent effects leads to another prediction: living systems in complex environments will evolve context and attention switching systems. On a classical view, this entails the identification of context (cf. the role of pointers) in a hierarchical generative model, where high level states contextualise the processing of lower level states. This immediately introduces the notion of MBs or holographic screens in the joint space of an agent's internal states, i.e. a notion of modularity supported by shared memory.  This is gracefully accommodated by the channel theory of QRFs as developed in \S\ref{chan-th}.\footnote{See \cite{fg:20b} for an explicit application to global-workspace theory and \cite{safron:20} for a more ambitious synthesis of active inference, a global workspace, and IIT.} Indeed, much work in the classical field of active inference rests upon optimising the hierarchical structure of deep generative models, via various free energy minimizing processes \cite{davis:03, george:09, friston:17a, friston:17b}. This is especially true for generative models of navigation and language \cite{mackay:95b, teh:06, friston:20} that are almost universally based upon quantisation or discretisation of state spaces. The attendant Boolean logic that arises from the imposition of Boolean constraints by QRFs may have important practical applications, as demonstrated by recent work in active vision and scene construction \cite{parr:17} that rests upon a generative model of the sensory (visual) consequences of visually foraging a scene with multiple objects. Active inference in this context can be seen as an inversion of a generative model that maps from causes (external objects), to consequences (sensory states) (cf. \cite{pizlo:01}).

The inversion of generative models for active vision is extremely difficult and ill-posed due to its computational complexity. These models have to accommodate all the natural physical laws of motion and optics (e.g., occlusion). The inversion of these models corresponds to the measurement operators that mediate belief updating. In a classical setting, this can be massively finessed by coarse graining the problem and applying Boolean operators. For example, if one object is near and another object is far, an agent will see the near object. Quantising or coarse graining the internal (i.e. generative) model along these lines reduces the likelihood mapping to a set of relatively simple Boolean operators that could be cast as measurement operators in a quantum-theoretic context. At present, when these operators are encoded on standard classical computers for simulation, they are effectively implemented with enormous tensors. Reading and writing these tensors into working memory dwarfs the actual compute time and renders active vision schemes of this sort computationally and thermodynamically inefficient. One might imagine that a combination of quantum computing \cite{feynman:82, nielsen:00} and neuromorphic engineering \cite{mead:90, tang:19} may be able to parallel the efficiency of human vision.

There are, in addition, certain technical problems that are resolved in moving from a classical to a quantum FEP formulation of active inference that go beyond a commitment to unitary processes and binary measurement operators. These include problems entailed by assuming reference frames that can be specified to infinite precision. In formalising the asymptotic behaviour of belief updating in the absence of random fluctuations in the classical formalism, for example, one has to deal with differential entropies that are ill-defined for very precise conditional densities, e.g. Dirac Delta functions. One workaround is to use Jaynes' limiting density of discrete points (LDDP) \cite{jaynes:57}, which brings us back to where we started, namely finite dimensional Hilbert spaces and discrete state spaces.

Finally, from a more philosophical perspective, the framework presented here is consistent with \cite{fgl:21, levin:20, friston:20a} in supporting a panpsychist perspective on questions of agency, sentience, and cognition.   As our Definition \ref{agent-def} of agency illustrates, traditional binary categorizations of entities into those having agency, sentience, and cognition and those lacking them are replaced here by continua of ``interestingness'' along these dimensions.  Definition \ref{agent-def} assumes free choice, in particular of measurement basis.  The Conway-Kochen theorem \cite{conway:09} states that if {\em any} physical system is assumed to have free choice, then {\em all} physical systems must be assumed to have free choice.  ``Free choice'' in this case means behavior that is not determined by (cannot be fully predicted given) the events in the system's past lightcone.  Determination of behavior by events in the past lightcone is {\em local} determinism; such local determinism is fundamentally inconsistent with the global determinism implied by unitarity \cite{fields:13a}.  Hence what is interesting is the {\em extent} of choice.  As emphasized in \cite{friston:19}, interesting choices are implemented by internal processes.   An electron, for example, has internal states -- its charge and mass are not states of its MB -- but they are invariant, and so do not implement choices.  The internal states of rocks are not invariants -- changes in water content or radioactive decay can occur -- but they do not, in general, implement interesting choices from our human perspective.  Interesting choices require interesting sensations and actions, and hence significant internal energy flows as discussed above.  They require differential {\em use} of thermodynamic resources to deploy multiple QRFs that probe the observable environment in different ways.  Systems, including organisms, clearly do this to different extents; even among humans, the extent of variation is striking.  Consonant with a broadly-construed enactivism \cite{maturana:80, varela:91, anderson:03, froese:09}, we view such active exploration of the environment -- hence, active inference -- as indicating cognition and intelligence.  

\subsection{Predictions and open questions} \label{pred}

The results obtained here suggest that the field of quantum biology is far larger than has so far been explored; indeed they suggest that all biological systems are properly considered quantum systems and can be expected to employ quantum coherence as an information processing resource.  This suggestion is of course not novel, having been explored by many authors from philosophical, theoretical, and increasingly over the past two decades, empirical perspectives \cite{schrodinger:44, wigner:61, penrose:89, hameroff:96, tegmark:00, davies:04, hameroff:04, arndt:09, lambert:12, melkikh:15, brookes:17, mcfadden:18, marais:18, cao:20, kim:21}.  We have, however, shown here that it follows from general principles: quantum systems with sufficiently complex internal information flows can be expected to exhibit active inference, and hence to behave like organisms.  Indeed, quantum systems that do {\em not} display evident intelligence -- quantum systems that are trivial agents -- become an unusual special case.  Conversely, the results obtained here suggest that the concepts of active inference, agency, Bayesian satisficing, and cognition are applicable to generic quantum systems, and hence to physical systems across the board.  This turns the traditional question of the emergence of intentionality \cite{polanyi:68, rosen:86, ellis:05} on its head: by making agency a generic expectation, it makes the ``merely physical'' the special case demanding explanation.   It thus suggests that the traditional division between agency and mechanism is a hindrance, not a help, in the task of understanding the natural world.

We make, in particular, the following predictions:
\begin{enumerate}
\item The internal, molecular-scale dynamics of both prokaryotic and eukaryotic cells implement quantum information processing, i.e., make essential use of quantum coherence as a computational and memory resource.  This is consistent with recent results showing that the free energy budgets of biological systems across known phylogeny are orders of magnitude short of the resources required for purely classical computation at the molecular scale \cite{fl:21}.
\item Interacting biological systems trade off separability against entanglement and hence classical against quantum communication.  Interactions between biological systems from the molecular scale upwards can be expected to display quantum contextuality and violations of the Bell and Leggett-Garg inequalities.
\end{enumerate}
\noindent
These predictions remain largely untested, both for reasons of technical difficulty and due to a still-pervasive traditional view of macroscopic systems as ontically classical.  The mechanistic role of entanglement -- even in relatively well-established systems such as light harvesting -- remains subject to considerable debate \cite{marais:18, cao:20, kim:21}.  Quantum context switching has been detected in human subjects \cite{cervantes:18, basieva:19}, but whether it is properly considered ``quantum'' or merely ``quantum-like'' remains open to question \cite{khrennikov:15}.\footnote{Such ``quantum-like'' contextuality is claimed for a certain case of gene expression in \cite{basieva:11} where conflicting probabilities (quantum vs. classical) do not give rise to a consistently definable joint distribution, so suggesting a variant of Kolmogorov contextuality.}  Neither Bell nor Leggett-Garg inequality violations have been conclusively demonstrated with biological systems.

Setting technical difficulties aside, we suggest that coherence effects in biological settings may be systematically overlooked due to being dismissed as ``noise'' or ``random coincidence.''  Thermal noise in biological settings is well characterized; general environmental variation, including effects of signaling by other cells or organisms, is less straightforwardly modeled or controlled.  Experimental designs that explicitly test for coherence effects will, we expect, be required to test the above predictions.  Focused theoretical support for such designs is therefore necessary.  

Interactions involving multiple agents remain an open theoretical as well as experimental-design problem.  In the bipartite setting employed for the technical analysis above, every agent interacts with its entire surrounding, whether the bits transferred by this interaction are observed and processed or not (i.e. whether they are included in sector $E$ or $F$).  This does not change in a multi-party setting; each agent still interacts with its entire surround, identifying other agents (or not) via specific QRFs.  The ubiquitous assumption that inter-agent communication is classical, made in domains as disparate as cell-cell interaction and human natural language use, becomes problematic in this setting.  As discussed in \S\ref{asymptotic}, shared QRFs are required for fully-shared, counterfactual-supporting semantics, but they induce entanglement (see also the discussion of this point in \cite{fgm:21}).  The extent of shared semantics is not readily observable in living systems.  In a fully classical setting, generalised synchronisation (a.k.a., synchronisation of chaos) emerges when two free energy minimising `partners' observe each other \cite{frith:15}.  ``Perfect prediction'' of the partner's behavior may result, driven by classical learning of a shared generative model, implicitly resolving the hermeneutics problem in the communication context \cite{frith:13, frith:15a}.  Recognizing a particular communication partner, in this case, requires invoking a particular generative model, the correct model to predict that partner's behavior \cite{isomura:19}.  Such classical synchronization is not, however, robust against perturbation; altering one agent's model does not ``automatically'' alter the other, as would be expected if the models, i.e. the QRFs implemented by the two agents were entangled.  Both theoretical and experimental characterization of relevant QRFs will be required to assess the extent to which classical communication can be considered purely classical, and detemine where and how quantum coherence contributes to in-practice successful shared semantics.

Darwinian evolution can be viewed as a process of variation and selection of QRFs, and hence as an instance of the multiple-agents problem in which semantics is only partially shared.  Evolution can be given a natural description in the framework of the classical FEP \cite{friston:13, campbell:16, fl:20a, fl:20b}.  Variation and selection have been advanced as a model of decoherence under the rubric of quantum Darwinism \cite{zurek:03, zurek:09}; see \cite{campbell:10} for a discussion from the perspective of universal Darwinism.  The selection mechanism invoked by quantum Darwinism, however, assumes QRF sharing by multiple agents \cite{fields:10}; see \cite{fgm:21, fm:20} for discussion.  While the present results allow any evolutionary system coupled to a larger environment to be viewed as a Bayesian agent implementing active inference, a fully-satisfactory account of variation and selection within a quantum framework remains to be developed.  

Additional theoretical work is also needed to understand the relationships among the many distinct models of quantum contextuality that have been put forward, often with the use of quite different formal tools (e.g. \cite{abramsky:11,abramsky:12a,popescu:14,adlam:21}).  The question of ``context'' is deeply tied up with that of what is to be regarded as the ``environment'' or ``surrounding'' of any given system.  This is a particularly critical question at the cellular level, where multiple signaling modalities with different spatial ranges and temporal characteristics are present.  Models that explicitly characterize the ``spaces'' in which cells and multicellular systems operate -- e.g. the space of potential morphologies, or that of potential messages from interaction partners -- and that consider constraining effects of one space on another both within and between scales will be needed to understand the roles of context, and of context switching, in biological systems.

In closing, we hope that we have shown to readers familiar with the FEP and the active inference framework that quantum effects are worth considering both theoretically and in experimental design.  For readers not familiar with the classical FEP but literate in quantum theory (or vice versa), we hope this paper has gone some way to contextualizing your QRFs in sense-making via Markov blankets and their underlying holographic screens.

\section*{Conflict of Interest Statement}

The authors declare that the research was conducted in the absence of any commercial or financial relationships that could be construed as a potential conflict of interest.

\section*{Funding}

The work of C.F. is supported in part by the Emerald Gate Foundation.  K.J.F. is supported by funding for the Wellcome Centre for Human Neuroimaging (Ref: 205103/Z/16/Z) and the Canada-UK Artificial Intelligence Initiative (Ref: ES/T01279X/1).  M.L. gratefully acknowledges support by the Guy Foundation Family Trust (103733-00001), the John Templeton Foundation (62212), and the Elisabeth Giauque Trust.

\end{document}